%% file: main.tex
\documentclass[11pt]{article}
\def\withcomments{0}
\input{preamble}
\usepackage[style=numeric,maxbibnames=99]{biblatex}
\newcommand{\email}[1]{\href{mailto:#1}{#1}}
\algnewcommand{\LineComment}[1]{\State \(\triangleright\) #1}

\newcommand{\outab}{K} %

\title{Uniformity Testing in the Shuffle Model:\\Simpler, Better, Faster} %
\author{Cl\'ement L. Canonne\thanks{University of Sydney. Email: \email{clement.canonne@sydney.edu.au}.} \and Hongyi Lyu\thanks{University of Melbourne. Email: \email{hongyi.lyu@student.unimelb.edu.au}.}}
\date{June 2021}

\begin{document}

\maketitle

\begin{abstract}
    Uniformity testing, or testing whether independent observations are uniformly distributed, is \emph{the} prototypical question in distribution testing. Over the past years, a line of work has been focusing on uniformity testing under privacy constraints on the data, and obtained private and data-efficient algorithms under various privacy models such as central differential privacy (DP), local privacy (LDP), pan-privacy, and, very recently, the shuffle model of differential privacy.
    
    In this work, we considerably simplify the analysis of the known uniformity testing algorithm in the shuffle model, and, using a recent result on ``privacy amplification via shuffling,'' provide an alternative algorithm attaining the same guarantees with an elementary and streamlined argument.
\end{abstract}

\section{Introduction}
Learning from, or, more generally, performing statistical inference on sensitive or private data has become an increasingly important topic, where one must balance the desire to achieve good accuracy with the requirement to preserve privacy of the users' data. Among the many tasks concerned, hypothesis testing and, more specifically, \emph{goodness-of-fit testing} is of particular importance, given its ubiquitous role in data analysis, the natural sciences, and more broadly as a workhorse of statistics and machine learning.

In this paper, we consider the specific case of \emph{uniformity testing}, the prototypical example of goodness-of-fit testing of discrete distributions, where one seeks to decide whether the data is drawn uniformly from a known finite domain. Investigating the trade-off between accuracy (or, equivalently, data requirements) and privacy for this task has received considerable attention over the past years in a variety of privacy models, including the central and local models of differential privacy, the so-called pan-privacy, and the recently proposed model of \emph{shuffle privacy}. 

Unfortunately, while this trade-off is now well understood in most of the aforementioned privacy settings, some of the proposed algorithms remain relatively complex and far from practical, and their analysis quite involved. With this in mind, we focus in this paper on private uniformity testing in the shuffle model, both simplifying the analysis of the existing algorithms for this task and obtaining a new, arguably simpler one with the same guarantees.

\subsection{Previous work}
    \label{ssec:previous}
Testing uniformity of discrete distributions was first considered from the theoretical computer science viewpoint in~\cite{GoldreichR00}, and the optimal sample complexity $\Theta(\sqrt{\ab}/\dst^2)$ obtained in~\cite{Paninski08}~--~where $\ab$ denotes the domain size and $\dst$ the distance parameter.\footnote{See~\cref{2} for a formal definition of the uniformity testing question and the different privacy models considered.} Over the past years, several followup works refined this result, for instance to generalise it to \emph{identity} testing (reference distribution other than uniform)~\cite{BatuFFKRW01,ValiantV17,AcharyaDK15,DiakonikolasKN15,DiakonikolasK16,Goldreich16} or to pinpoint the optimal dependence on the error probability~\cite{HuangM13,DiakonikolasGPP18}.

The question was then revisited from the privacy perspective by~\cite{CaiDK17}, after which~\cite{AcharyaSZ18,AliakbarpourDR18} established the tight sample complexity bound 
$\Theta(\sqrt{\ab}/\dst^2 + \sqrt{\ab}/(\dst\sqrt{\priv}) + {\ab^{1/3}}/(\dst^{4/3}\priv^{2/3}) + 1/(\dst\priv))$
under (central) differential privacy. In the more stringent locally private setting, the question was later raised in~\cite{Sheffet18}, and fully answered in a sequence of works which show that the tight sample complexity is $\Theta(\ab^{3/2}/(\dst^2\priv^2))$ for non-interactive private-coin protocols and $\Theta(\ab/(\dst^2\priv^2))$ for non-interactive public-coin protocols (where the users have access to a common random seed)~\cite{AcharyaCFT19,AcharyaCT20a,AcharyaCFST21}, and that allowing interactive protocols does not improve the sample complexity beyond $\Theta(\ab/(\dst^2\priv^2))$~\cite{AminJM20,BerrettB20,AcharyaCLST21}.

Focusing on the different model of \emph{pan-privacy},~\cite{AminJM20} establishes a sample complexity upper bound of $O(\ab^{2/3}/\dst^{4/3}\priv^{2/3})+\sqrt{\ab}/\dst^2+\sqrt{\ab}/(\dst\priv))$, as well as a near-matching lower bound of $\Omega(\ab^{2/3}/\dst^{4/3}\priv^{2/3})+\sqrt{\ab}/\dst^2+1/(\dst\priv))$. Interestingly, the leading dependence on the domain size $\ab$, while sublinear, is significantly greater than the $\sqrt{\ab}$ dependence in the central DP case, thus placing the ``cost of privacy'' in the pan-private model strictly in-between central and local privacy.

Finally, the work most relevant to ours,~\cite{BalcerCJM21}, tackles uniformity testing in the shuffle model of privacy. Building on a connection between \emph{robust} shuffle private algorithms and pan-private algorithms, they use ideas from~\cite{AminJM20} to derive both upper and lower bounds on the sample complexity of shuffle private uniformity testing. Specifically, they provide a \emph{private-coin}, \emph{robust}, \emph{approximate}-DP algorithm for uniformity testing in the shuffle model with sample complexity
\begin{equation}
    \label{eq:sc:privatecoin}
O\Paren{\frac{\ab^{3/4}}{\dst\priv}\log^{1/2}\frac{\ab}{\delta} +
\frac{\ab^{2/3}}{\dst^{4/3}\priv^{3/3}}\log^{1/3}\frac{\ab}{\delta} +
\frac{\sqrt{\ab}}{\dst^2} }\,.
\end{equation}
Leveraging, as in~\cite{AminJM20}, the ``domain compression'' technique from~\cite{AcharyaCFT19,AcharyaCT20b}, one can then derive a \emph{public-coin}, robust, approximate-DP testing algorithm with sample complexity
\begin{equation}
    \label{eq:sc:publiccoin}
O\Paren{\Paren{\frac{\ab^{2/3}}{\dst^{4/3}\priv^{2/3}} + \frac{\sqrt{\ab}}{\dst\priv} + \frac{\sqrt{\ab}}{\dst^2}}\log^{1/2}\frac{\ab}{\delta} }\,.
\end{equation}
They conclude by deriving a sample complexity lower bound for public-coin, robust, \emph{pure}-DP testing algorithms of
\begin{equation}
    \label{eq:sc:publiccoin:lb}
\Omega\Paren{\frac{\ab^{2/3}}{\dst^{4/3}\priv^{2/3}} + \frac{1}{\dst\priv} + \frac{\sqrt{\ab}}{\dst^2}}\,.
\end{equation}
We refer the reader to~\cite{Canonne20,Cheu21} and references therein for more on the distribution testing literature, both without and with privacy constraints.

\subsection{Our contributions} %

In this paper, we focus on private-coin, robust shuffle private algorithms for uniformity testing; specifically, on the upper bound~\eqref{eq:sc:privatecoin} established in~\cite{BalcerCJM21}, which as discussed above implies, in a black-box fashion, the public-coin upper bound~\eqref{eq:sc:publiccoin}. However, the proof of~\eqref{eq:sc:privatecoin} from~\cite{BalcerCJM21} is quite unwieldy, and the analysis involves breaking the proposed test statistic in 4 parts, before dealing which each of them separately, invoking intermediary lemmas (from previous work, and new ones) in order to handle each part. As a result, the argument involves many moving parts and seemingly arbitrary constants, and results in an algorithm which, while theoretically sound, does not appear intuitive nor practical.

Our contribution is fourfold. We first revisit the algorithm of~\cite{BalcerCJM21}, making a minor modification to the analyser's side; in turn, this minor modification, combined with a simple observation (\cref{Binomial of Poisson}), lets us considerably streamline and simplify the analysis of the algorithm. At a high level, this observation implies that all the resulting random quantities considered, in spite of the complicated process that generates them, \emph{turn out to have a very simple and manageable distribution}~--~namely, they are Poisson distributed. This follows from the aforementioned~\cref{Binomial of Poisson}, along with basic properties of Poisson distributions such as stability under convolutions. This in turn lets us easily bound the expectation and variance of our test statistic, by relying on the known moments of Poisson random variables.

Having streamlined the analysis of the algorithm, our second contribution is to simplify and slightly improve its statement. Indeed, our approach immediately allows us to shave the spurious logarithmic dependence on $\ab$ (while keeping the same robust shuffle privacy guarantees). We also note that out of the three terms of the sample complexity stated in~\eqref{eq:sc:privatecoin}, the middle term is always dominated by the other two, and thus can be removed. This leads to our sample complexity
\begin{equation}
    \label{eq:sc:privatecoin:better}
O\Paren{\frac{\ab^{3/4}}{\dst\priv}\sqrt{\log\frac{1}{\delta}} + \frac{\sqrt{\ab}}{\dst^2} }
\end{equation}
for private-coin, robust shuffle private uniformity testing.\footnote{We emphasise that this simplified sample complexity is not our main contribution, which lies in the improved analysis of the algorithm.}

Our third contribution is a \emph{new} private-coin shuffle algorithm, achieving the same sample complexity~\eqref{eq:sc:privatecoin:better}, privacy, and robustness as the previous one through a totally different approach. Namely, we employ the ``privacy amplification by shuffling'' paradigm~\cite{ErlingssonFMRTT19} to convert a \emph{locally private} uniformity testing algorithm to a robust shuffle private testing algorithm. In order to do so, we leverage the recent result of~\cite{FeldmanMT20}, which provides an amplification-by-shuffling theorem optimal in all privacy regimes~--~quite fortunately, as an optimal dependence on $\priv$ even in the low-privacy regime is exactly what we need to achieve~\eqref{eq:sc:privatecoin:better}.

This leads us to our fourth and last contribution: a locally private uniformity testing. Although optimal LDP testing algorithms are known in the high-privacy regime~\cite{AcharyaCFT19,AcharyaCFST21}, for our LDP-to-shuffle conversion to go through we need to start from an LDP testing algorithm with very low privacy (large $\priv$). Building on ideas from~\cite{AcharyaSZ19}, we adapt the algorithm from~\cite{AcharyaCFT19} to this low-privacy setting, obtaining as a byproduct a locally private uniformity testing algorithm with optimal sample complexity in all parameter regimes.

To conclude this discussion, we note that using the same ``domain compression'' method as~\cite{BalcerCJM21} (see \cref{theo:random:dct:hashing} for the statement) in a blackbox fashion, both our private-coin algorithms immediately imply simple \emph{public-coin} robust shuffle private algorithms with sample complexity~\eqref{eq:sc:publiccoin}; actually, the slightly better\footnote{This complexity being, up to the absence of $\ab$ in the logarithm, what~\cite{BalcerCJM21} actually proved, if one looks at their derivation carefully to improve some inequalities.}
\begin{equation}
    \label{eq:sc:publiccoin:better}
O\Paren{\frac{\ab^{2/3}}{\dst^{4/3}\priv^{2/3}}\log^{1/3} \frac{1}{\delta} + \frac{\sqrt{\ab}}{\dst\priv}\log^{1/2} \frac{1}{\delta} + \frac{\sqrt{\ab}}{\dst^2}}\,.
\end{equation}
Therefore, our simplifications and our new algorithm also end up applying to the public-coin result, by combining them with domain compression (which is the only part relying on public randomness).

\subsection{Organization}  %
Basic definitions are provided in Preliminaries (\cref{2}), along with a few relevant theorems (either from previous works or ``folklore''). We then revisit the algorithm presented in~\cite{BalcerCJM21} for lower bounds in uniformity testing to both simplify the argument and improve the bound (\cref{3}). In addition, we propose an alternative approach to obtain the same improved bound in \cref{4}, based on a new \emph{locally private} testing algorithm and the ``amplification by shuffling'' paradigm.
\section{Preliminaries}
\label{2}
Throughout, we use the standard asymptotic notation $O(\cdot)$, $\Omega(\cdot)$, $\Theta(\cdot)$, and denote by $\log$ the natural logarithm. Given an integer $\ab \in \N$, we write $[\ab]$ for the set $\{1,2,\dots,\ab\}$, and $\bU_\ab$ for the uniform distribution over $[\ab]$. When clear from context, we will drop the subscript and simply write $\bU$. Finally, given two tuples $x,y\in \cY^\ast$, we write $x\sqcup y\in\cY^\ast$ for their concatenation.

\paragraph{Differential privacy.} We formally define the three notions of privacy central to this paper. Two $n$-elements datasets $\vec{x},\vec{x}'\in\cX^n$ are said to be \emph{adjacent} if they differ in at most one element, i.e., if their Hamming distance is at most one.

\begin{definition}[Differential Privacy~\cite{DworkMNS06}]
Fix $\priv>0$ and $\delta\in[0,1]$. A (randomised) algorithm $\cM\colon\cX^n\to \cZ$ is said to be \emph{$(\priv,\delta)$-differentially private} (DP) if
\[
    \proba{\cM(\vec{x})\in S} \leq e^\priv \proba{\cM(\vec{x}')\in S} +\delta
\]
for all $S\subseteq \cZ$ and all adjacent $\vec{x},\vec{x}'\in\cX^n$.
\end{definition}
This definition underpins all the privacy models we consider; the difference between them lies in which part of the protocol is required to be differentially private (the individual messages from the users, or the tuple of all users' messages after random permutation), as illustrated in~\cref{fig:dps}. Moreover, one can also consider two types of distributed protocols, depending on whether the users share some common random seed (which is independent of their data, and as such not subject to privacy requirements), which could in some case allow them to achieve better accuracy. Following the standard terminology, these two types of protocols are referred to as \emph{public-} and \emph{private-coin} protocols (the tossed ``coins'' being the randomness available): importantly, the ``private'' in private-coin protocols does not refer to (differential) privacy, but to the randomness available to each user being independent from the others'.
\begin{definition}[Local Differential Privacy~\cite{DworkMNS06,KLNRS11}]
A (non-interactive) protocol $\Pi$ for $n$ users in the \emph{local model} consists of the following:
\begin{itemize}
    \item $n$ \emph{randomisers} $\cR_1,\dots, \cR_n\colon \cX\times \{0,1\}^r\to\cY$, each mapping a data point $x$ and public random bits $u$ to a message $\cR_i(x,u)$;
    \item an \emph{analyser} $\cA\colon \cY^n\times \{0,1\}^r \to \cZ$, which takes as input the $n$ messages and the public random bits
\end{itemize}
(all randomised mappings). On input $\vec{x}\in\cX^n$, the output of $\Pi$ is $\cA(\cR_1(\vec{x}_1, U), \dots, \cR_n(\vec{x}_n, U), U)$, where $U$ is uniform on $\{0,1\}^r$. If $r=0$ (no public randomness), $\Pi$ is a \emph{private-coin} protocol; otherwise, it is a \emph{public-coin} protocol. For $\priv>0$ and $\delta\in[0,1]$, $\Pi$ is said to be \emph{$(\priv,\delta)$-locally differentially private} (LDP) if each $\cR_i$ is $(\priv,\delta)$-DP.
\end{definition}

\begin{definition}[Shuffle Differential Privacy~\cite{BittauEMMRLRKTS17,CheuSUZZ19}]
A  (non-interactive) shuffle protocol $\Pi$ for $n$ users consists of 
\begin{itemize}
    \item a \emph{randomiser} $\cR \colon \cX\times \{0,1\}^r  \to \cY^\ast$, which maps a data point $x$ and public random bits $u$ to a  (possibly variable-length) tuple of messages;
    \item a \emph{shuffler} $\cS\colon \cY^\ast \to \cY^\ast$, which concatenates $n$ tuples of messages and returns a uniformly random permutation of the resulting tuple; and
    \item an \emph{analyser} $\cA\colon \cY^\ast\times \{0,1\}^r \to \cZ$, which takes as input the permuted tuple of messages and the public random bits
\end{itemize} 
(all randomised mappings). On input $\vec{x}\in\cX^n$, the output of $\Pi$ is 
\[
    \cA( \cS\circ\cR^n(\vec{x},U), U)\,,
\]
where $U$ is uniform on $\{0,1\}^r$ and
\[
    \cS\circ\cR^n(\vec{x},U) \eqdef \cS(\cR(\vec{x}_1, U)\sqcup \dots \sqcup\cS(\cR(\vec{x}_n, U))
\]
is the permuted tuples of messages output by the shuffler. If $r=0$ (no public randomness), $\Pi$ is a \emph{private-coin} protocol; otherwise, it is a \emph{public-coin} protocol. For $\priv>0$ and $\delta\in[0,1]$, $\Pi$ is said to be \emph{$(\priv,\delta)$-shuffle differentially private} (shuffle DP) if the output of the shuffler $\cS\circ\cR^n(\cdot, u)$ is $(\priv,\delta)$-DP for every fixed $u\in\{0,1\}^r$.
\end{definition}
While a bit cumbersome to parse, this definition captures the idea of only requiring the set of messages from the users to be differentially private \emph{after} being randomly shuffled, in contrast to asking this directly of the set of messages (in the more stringent LDP setting) or only of the analyser's output (in the more permissive DP setting). See~\cref{fig:dps} for an illustration.
\begin{figure}[htp]
    \centering
    \includegraphics[height=0.26\textwidth]{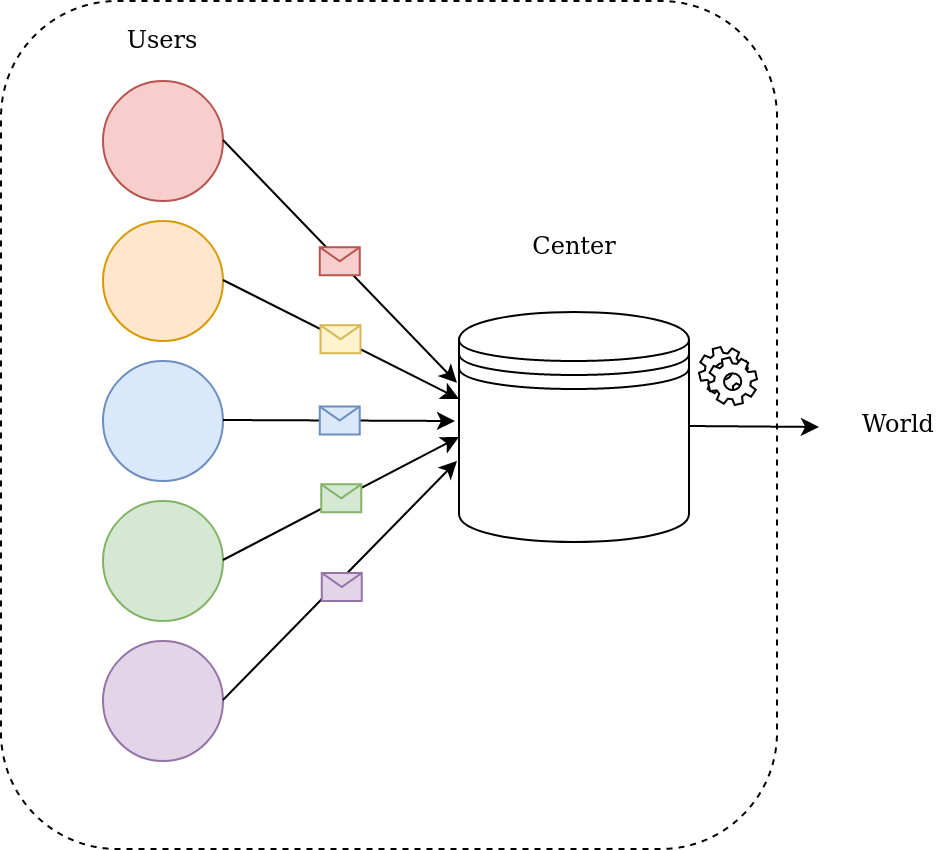}\hfill
    \includegraphics[height=0.26\textwidth]{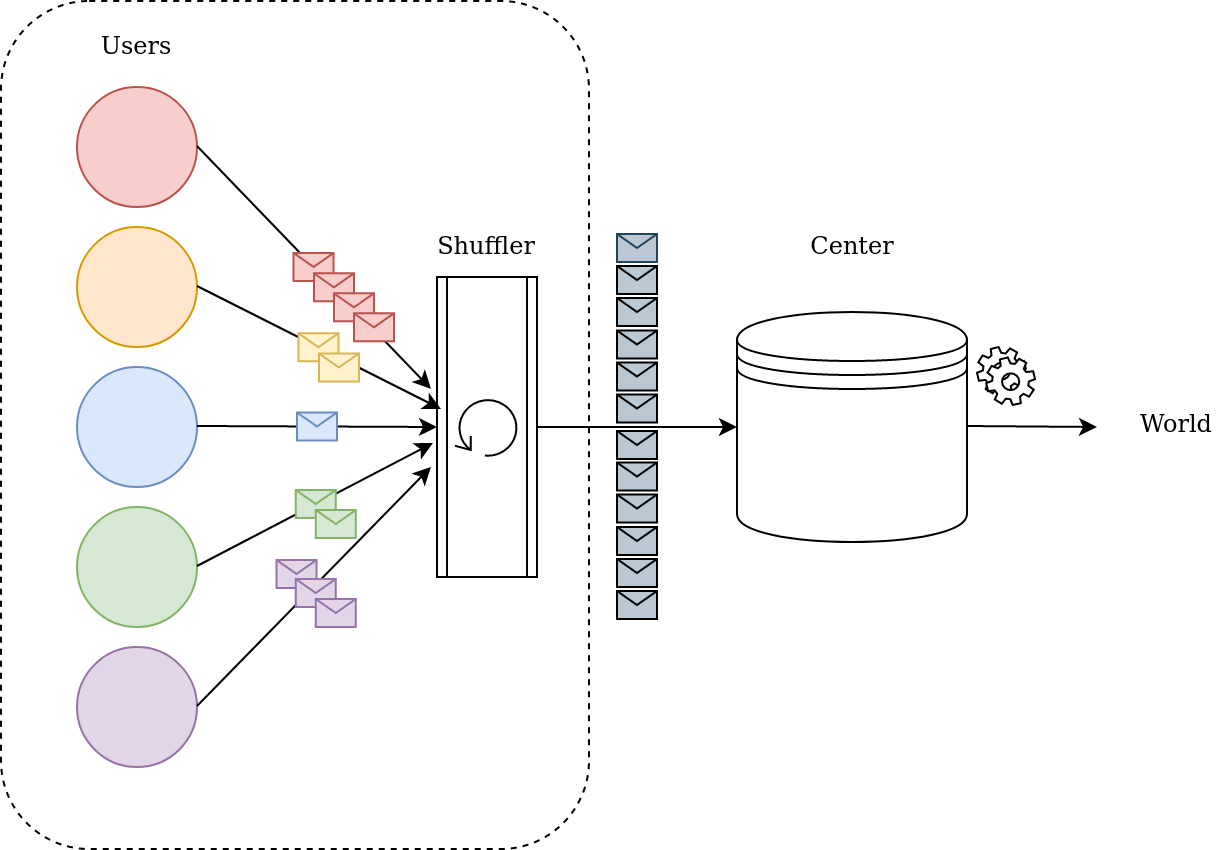}\hfill
    \includegraphics[height=0.26\textwidth]{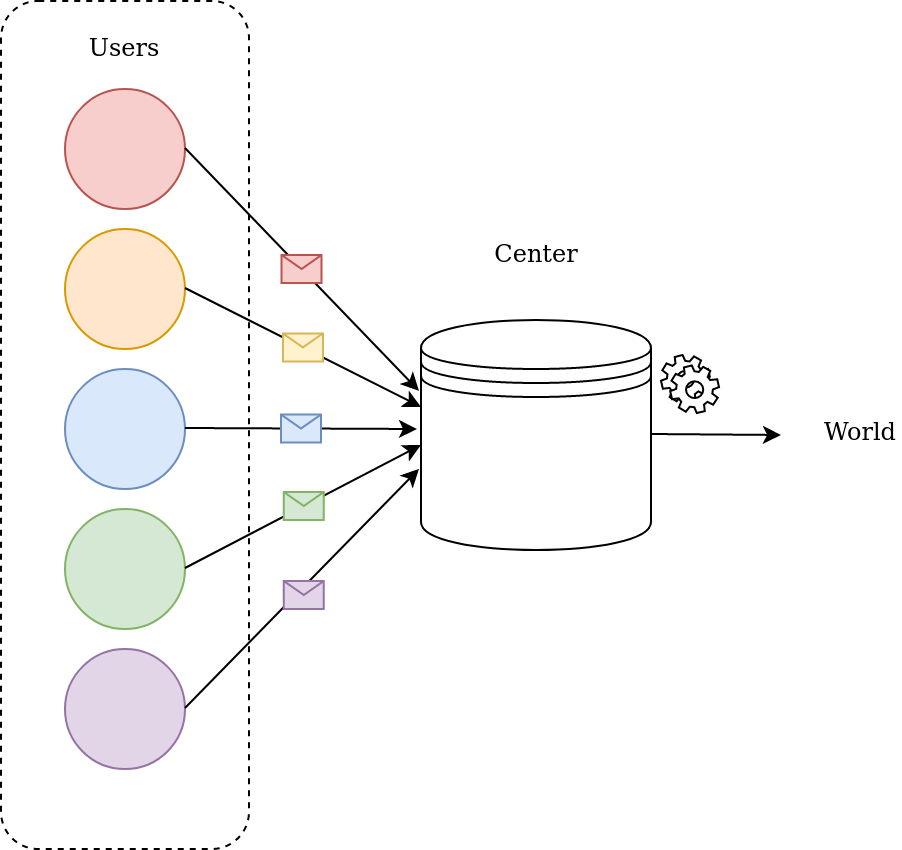}
    \caption{An illustration of the DP (left), shuffle DP (middle), and LDP (right) settings. In all cases, the dotted box denotes which part of the system must satisfy $(\priv,\delta)$-differential privacy.}
    \label{fig:dps}
\end{figure}

Finally, we will require one last variant of the shuffle DP setting, which guarantees privacy even if a subset of the users acts maliciously.
\begin{definition}[Robust Shuffle Differential Privacy~\cite{BalcerCJM21}]
A  (non-interactive) shuffle protocol $\Pi=(\cR,\cS,\cA)$ for $n$ users is \emph{robustly} shuffle DP if there exist two non-increasing, continuous functions $\bar{\priv} >0, \bar{\delta}\in[0,1]$ such that $\cS\circ\cR^{\gamma \sa}(\cdot, u)$ is $(\bar{\priv}(\gamma),\bar{\delta}(\gamma))$-DP for every $u$, and every $\gamma \in(0, 1]$.
\end{definition}
Loosely speaking, the robustness requirement ensures that the shuffle protocol still provides meaningful privacy even if some of the $\sa$ users depart from the protocol; in which case the privacy guarantee holds for the remaining honest users, and degrades smoothly with their number $\gamma \sa$. 
We refer the reader to~\cite{DworkR13} for more detail on differential privacy and local privacy; and to~\cite{Cheu21} for further discussion of the interplay between those notions and shuffle (resp., robust shuffle) privacy. 

\paragraph{Distribution testing.} We now define the main focus of this paper, \emph{uniformity testing.}

\begin{definition}[Uniformity Testing]
Let $\p$ be an unknown distribution over $[\ab]$. A \emph{uniformity testing algorithm with sample complexity $n$} takes inputs $\dst\in (0,1]$ and a set of $n$ i.i.d.\ samples from $\p$ and outputs either ``\emph{accept}'' or ``\emph{reject}'' such that the following holds.
\begin{itemize}
\item If $\p=\bU$, then the algorithm outputs "\emph{accept}" with probability at least 2/3;
\item If $\totalvardist{\p}{\bU}>\dst$, then the algorithm outputs "\emph{reject}" with probability at least 2/3,
\end{itemize}
where $\totalvardist{\p}{\q}\eqdef \sup_{S \subseteq [\ab]} (\p(S)-\q(S)) = \frac{1}{2}\norm{\p-\q}_1 \in [0,1]$ is the total variation distance between arbitrary distributions $\p$ and $\q$ over $[\ab]$.
\end{definition}
A few remarks are in order. First, the error probability of the algorithm has been set to $1/3$ above for simplicity, as common in the literature. It can be decreased to an arbitrary constant $\beta \in (0,1]$ at the cost of a multiplicative $\log(1/\beta)$ factor in the sample complexity (by repeating the test independently and taking the majority outcome).

Second, for ease of exposition, and consistency with previous work, we focus on uniformity testing, i.e., testing with null hypothesis being the uniform distribution $\bU$. However, our results readily extend to the more general task of \emph{identity} testing, where the null hypothesis is an arbitrary (known) reference distribution $\q$. This extension follows either by a straightforward modification of our proofs, or using the known reduction between uniformity and identity testing (cf.~\cite{DiakonikolasK16,Goldreich16}, or~\cite[Appendix~D]{AcharyaCT20b}).

Finally, as often in the distribution testing literature, we will use the so-called ``Poissonisation trick'' and prove the algorithm with a random number $\Pois(n)$ of samples instead of exactly $n$. This modification simplifies some aspects of the analysis, by introducing independence between the number of occurrences of each domain element; and can be done without loss of generality, due to the strong concentration of the Poisson distribution around its mean. We refer the reader to~\cite[Appendix~D.3]{Canonne20} for a short primer on Poissonisation.

\paragraph{Useful tools from previous work (or ``folklore'').} Finally, we recall three key results which our analysis will rely on. The first states that adding properly calibrated Poisson noise to an integer-valued function ensures privacy of the output.\footnote{We note that~\cite{GhaziKMP21} states the theorem with $16\log(10/\delta)$ instead of $16\log(2/\delta)$; but their proof establishes the $16\log(2/\delta)$ bound.}
\begin{lemma}[{Poisson Mechanism~\cite[Theorem~11]{GhaziKMP21}}]
\label{Poisson Mechanism}
Let $f\colon\cX^n\to\Z$ be a $\Delta$-sensitive function, i.e., such that $|f(\vec{x})-f(\vec{x}')|\leq \Delta$ for all neighbouring datasets $\vec{x},\vec{x}'\in\cX^n$. For any $\priv>0$, $\delta\in(0,1]$, and $\lambda \in\R$ such that
\[
    \lambda \geq \frac{16\log(2/\delta)}{(1-e^{-\priv/\Delta})^2} + \frac{2\Delta}{1-e^{-\priv/\Delta}},
\]
the algorithm which, on input $\vec{x}\in\cX^n$, samples $\eta\sim \Pois(\lambda)$ and outputs $f(\vec{x})+\eta$ is $(\priv,\delta)$-DP.
\end{lemma}
In particular, note that for $\Delta=1$ and $\priv\in(0,1]$ it suffices to have $\lambda = O(\log(1/\delta)/\priv^2)$. The second fact states that summing a Poisson number of i.i.d.\ Bernoulli random variables ends up simply being a Poisson random variable; we provide a proof for completeness.
\begin{lemma}[Binomial of Poisson is Poisson]
\label{Binomial of Poisson}
Let $\lambda \geq 0$, $p\in[0,1]$. Suppose $(X_i)_{i=1}^\infty$ are i.i.d.\ Bernoulli random variables with parameter $p$, and $N$ is a $\Pois(\lambda)$ random variable independent of the $X_i$'s. Then $\sum_{i=1}^N X_i \sim \Pois(\lambda p)$.
\end{lemma}
\begin{proof}
By definition, $X \eqdef \sum_{i=1}^N X_i = \sum_{i=1}^\infty X_i \mathbf{1}_{N \geq i}$; we can then compute the moment-generating function (MGF) of $X$ as follows, using the expressions of the MGF of a Bernoulli and Poisson distributions, respectively. For $t\in\R$,
\begin{align*}
    \ex{}{e^{t X}} 
    &= \ex{}{\prod_{i=1}^\infty e^{t X_i \mathbf{1}_{N \geq i}}}
    = \ex{}{\condex{}{\prod_{i=1}^\infty e^{t X_i \mathbf{1}_{N \geq i}} }{N}}
    = \ex{}{\prod_{i=1}^\infty \condex{}{e^{t X_i \mathbf{1}_{N \geq i}} }{N}}
    = \ex{}{\prod_{i=1}^N \ex{}{e^{t X_i}}} \\
    &=\ex{}{(1+ p(e^t-1))^N}
    = e^{\lambda((1+ p(e^t-1))-1)}
    = e^{\lambda p(e^t-1)}\,,
\end{align*}
in which we recognize the MGF of a $\Pois(\lambda p)$ random variable.
\end{proof}

Finally, we recall the domain compression technique to reduce sample complexity using public randomness.
\begin{lemma}[Domain Compression Lemma~\cite{AcharyaCT20b,AcharyaCHST20,AminJM20}]
  \label{theo:random:dct:hashing}
  There exist absolute constants $c_1,c_2>0$ such that the following holds. For any $2\leq \ell\leq \ab$ and any distributions $\p,\q$ on $[\ab]$,
  \[
        \proba[\Pi]{ \totalvardist{\p_\Pi}{\q_\Pi} \geq c_1\sqrt{\frac{\ell}{\ab}}\totalvardist{\p}{\q} } \geq c_2\,,
  \] 
  where $\Pi=(\Pi_1,\dots\Pi_\ell)$ is a uniformly random partition of $[\ab]$ in $\ell$ subsets, and $\p_\Pi$ denotes the probability distribution on $[\ell]$ induced by $\p$ and $\Pi$ via $\p_\Pi(i) = \p(\Pi_i)$.
\end{lemma}
\noindent This lets users leverage public randomness to agree on a common random partition of the domain and solve the testing problem on this new domain, trading a smaller domain size (from $\ab$ to $\ell$) for a smaller distance parameter (from $\dst$ to $\dst\sqrt{\ell/\ab}$).
\section{First algorithm: streamlining the argument of~\cite{BalcerCJM21}}
\label{3}

In~\cite{BalcerCJM21}, it was proved that uniformity testing could be performed in a robustly shuffle private manner, with a private-coin protocol. Specifically, the authors established the following:
\begin{theorem}[{\cite[Theorem~4.1]{BalcerCJM21}}]
\label{thm: previous bound}
	Let $\gamma \in (0,1]$, $\priv > 0$, and $\dst, \delta \in (0,1]$. There exists a private-coin protocol which is $(2\priv, 8\delta^\gamma,\gamma)$-robustly shuffle private and, for $\priv = O(1)$ and $\delta = o(1)$, solves $\dst$-uniformity testing with sample complexity
	\[
	\sa = O\Paren{ \frac{\ab^{3/4}}{\dst \priv} \log^{1/2} \frac{\ab}{\delta} + \frac{\ab^{2/3}}{\dst^{4/3}\priv^{2/3}} \log^{1/3} \frac{\ab}{\delta} + \frac{\ab^{1/2}}{\dst^2}  }.
	\]
\end{theorem}

We first observe that the \emph{statement} can be simplified. Indeed, the second term given in the sample complexity is redundant, as it is always dominated by one of the other two. To see why, suppose that the middle term dominates the first, i.e., $\ab^{\frac{2}{3}}\priv^{-\frac{2}{3}}\dst^{-\frac{4}{3}}\log^{\frac{1}{3}}(\ab/\delta)>\ab^{\frac{3}{4}}\priv^{-1}\dst^{-1}\log^{\frac{1}{2}}(\ab/\delta)$. Then $\dst^{-\frac{1}{3}} > \ab^{\frac{1}{12}} \priv^{-\frac{1}{3}}\log^{\frac{1}{6}}(\ab/\delta)$, which implies that $\dst^{-\frac{2}{3}} > \ab^{\frac{1}{6}} \priv^{-\frac{2}{3}}\log^{\frac{1}{3}}(\ab/\delta)$. That is, $\ab^{\frac{1}{2}}\dst^{-2} > \ab^{\frac{2}{3}}\priv^{-\frac{2}{3}}\dst^{-\frac{4}{3}}\log^{\frac{1}{3}}(\ab/\delta)$, and thus the third term dominates.

We then greatly simplify the \emph{proof} of the result. That is, we will provide a significantly simpler analysis of the algorithm, which provides the same privacy guarantee and, as a byproduct, removes the logarithmic dependence on $\ab$.\footnote{Our argument also yields explicit constants which, even without trying to optimise them, are considerably smaller than the ones from the original proof.}\medskip

\noindent\textbf{General setup.} The setup is similar to that of~\cite{BalcerCJM21}: we will again employ the \emph{Poissonisation} trick by using $\rsa \sim \Pois(\sa)$ samples instead of $\sa$, so that the counts of every element in the sample is independent from each other; we assume that the resulting value of $\rsa$ is known to the users and analyser. Each user will receive one sample from $\p$ and the samples are i.i.d. Let $\vec{x}$ denote the vector of inputs from users and $c_j(\vec{x})$ the actual count of $j \in [\ab]$ in $\vec{x}$. By Poissonisation, we have that $c_j(\vec{x}) \sim \Pois(\sa \p_j)$.

The randomiser $\mathcal{R}_{UT}$ (\cref{algo:1}) will be the same as that in \cite{BalcerCJM21} but our analyser (\cref{algo:2}) will be slightly different (modified for ease of computation). They are both given below, together constituting our private-coin robustly shuffle private uniformity tester $\mathcal{P}_{UT}=(\mathcal{R}_{UT},\mathcal{A}_{UT})$.

\begin{algorithm}[H]
\begin{algorithmic}[2]
\Require User data point $x\in[\ab]$, parameter $\lambda>0$, number of users $\rsa~\sim\Pois(\sa)$
\State Initialise the vector of messages $\vec{y} \gets \emptyset$
\For{$j\in[k]$}
\LineComment{Append an ``informative'' message to the tuple of messages, $(j,0)$ or $(j,1)$ depending on $x$}
\State $\vec{y} \gets \vec{y}\sqcup \{(j,\ind{x=j})\}$
\LineComment{Choose a random number of ``noise'' messages to append, each randomly $(j,0)$ or $(j,1)$}
\State Draw $s_j \sim \Pois(\lambda/ \rsa)$
\For{$t \in [s_j]$}
\State Draw $b_{j,t} \sim \Ber{1/2}$
\State Append $(j,b_{j,t})$ to $\vec{y}$
\EndFor
\EndFor
\Return{$\vec{y}$}
\end{algorithmic}
\caption{$\cR_\ut$, a randomiser for private uniformity testing}
\label{algo:1}
\end{algorithm}

\begin{algorithm}[H]
\begin{algorithmic}[2]
    \Require A message vector $\vec{y}\in([\ab]\times\{0,1\})^\ast$, parameter $\lambda>0$, number of users $\rsa~\sim\Pois(\sa)$
    \State Compute the ``noisy reference'' $\mu \gets \frac{\sa}{\ab} + \frac{\lambda}{2}$
    \For{$j\in[\ab]$}
        \State Compute count of $j$ as $N_j(\vec{y}) \gets |\{ i : y_i = (j,1)\}|$
    \EndFor
    \State Compute statistic 
    \begin{equation}
        \label{eq:def:statistic:Z}
        Z \gets \frac{\ab}{\sa} \sum_{j=1}^\ab \mleft( (N_j(\vec{y}) - \mu)^2 - N_j(\vec{y}) \mright)
    \end{equation}
    \If{ $Z > 2\sa\dst^2$ }
        \Return ``not uniform''
    \Else\ 
        \Return ``uniform''
    \EndIf
\end{algorithmic}
    \caption{$\cA_{\rm{}UT}$, an analyser for private uniformity testing}
    \label{algo:2}
\end{algorithm}

We are now ready to state our result.
\begin{theorem}
\label{thm: improved bound}
Let $\gamma \in (0,1]$, $\priv > 0$, and $\dst, \delta \in (0,1)$. There exists $\lambda=\lambda(\priv,\delta) \in \R$ such that the protocol $\cP_\ut = (\cR_\ut,\cA_\ut)$ is $(2\priv, 4\delta^\gamma,\gamma)$-robustly shuffle private. For $\priv = O(1)$ and $\delta = o(1)$, $\cP_\ut$ solves $\dst$-uniformity testing with sample complexity
	\[
	  \sa  = O\Paren{ \frac{\ab^{3/4}}{\dst \priv} \log^{1/2}\frac{1}{\delta} + \frac{\ab^{1/2}}{\dst^2}  }
	\]
\end{theorem}

\begin{proof} We establish separately the privacy and accuracy guarantees, starting with the former. Specifically, we want to prove the \emph{robust} privacy statement; throughout, we set the parameter $\lambda$ to
\begin{equation}
    \label{eq:lambda}
    \lambda(\priv,\delta)\eqdef \frac{64\log(2/\delta)}{(1-e^{-\priv})^2} = O\Paren{\frac{\log(1/\delta)}{\priv^2}}
\end{equation}

\noindent\textbf{Privacy.} Fix any $\gamma \in(0,1]$. By the same argument as in~\cite[Theorem 4.2]{BalcerCJM21}, we first restrict our attention to a single element j and by composition, in order to prove $(2\priv,4\delta^{\gamma})$-privacy of $(\cS\circ\cR_\ut^{\gamma \rsa})$, it suffices to show $(\priv,2\delta^{\gamma})$-shuffle privacy of messages containing j. To that end we only need to show $(\priv,2\delta^{\gamma})$-DP of \[c_j(\vec{x})+\sum_{i=1}^{\gamma \rsa}\Bin\mleft(\Pois\mleft(\frac{\lambda}{\rsa}\mright), 1/2\mright) = c_j(\vec{x}) + \Bin\mleft(\Pois\mleft(\gamma \lambda\mright), 1/2\mright)\]
since each $(j,1)$ has distribution $\Pois(\frac{\lambda}{\rsa})$.

By~\cref{Binomial of Poisson}, we know that 
\[
    \Bin\mleft(\Pois\mleft(\gamma \lambda\mright), 1/2\mright) = \Pois\mleft(\gamma\lambda/2\mright) 
\]
By~\cref{Poisson Mechanism}, to achieve $(\priv(\gamma),\delta(\gamma))$-DP, it suffices to have
\[ 
\frac{\gamma\lambda}{2}\geq \frac{16\log(2/\delta(\gamma))}{(1-e^{-\priv(\gamma)})^2} + \frac{2}{1-e^{-\priv(\gamma)}}\;
\]
which is satisfied for $\delta(\gamma) \eqdef 2^{1-\gamma} \delta^{\gamma} \leq 2\delta^\gamma$ and $\priv(\gamma)=\priv$, since then
\begin{align*}
    \frac{2}{1-e^{-\priv(\gamma)}}+\frac{16\log(2/\delta(\gamma))}{(1-e^{-\priv(\gamma)})^2}
    &\leq 2\cdot \frac{16\log(2/\delta(\gamma))}{(1-e^{-\priv})^2} 
    =\frac{\gamma\lambda}{2}\,.
\end{align*}
This establishes the privacy part of the theorem.\smallskip

\noindent\textbf{Sample Complexity.} We will show that the statistic $Z$ is smaller than our threshold $\tau\eqdef 2\sa\dst^2$ with constant probability when $\p=\bU$; and $Z$ is greater than $\tau=2\sa\dst^2$ with constant probability when $\totalvardist{\p}{\bU}>\dst$. To do so, we will compute the expectation and variance of $Z$ in both cases, and show that the ``gap'' between the two expectations (which will be $2\tau$) is much larger than the standard deviations; this will allow us to conclude by Chebyshev's inequality.

We implement this roadmap in the next 3 lemmas, which bound the first two moments of $Z$.
\begin{lemma}[Expectation of the statistic]
\label{mean of Z}
    Let $Z$ be defined as in~\eqref{eq:def:statistic:Z}, when $\rsa\sim\Pois(\sa)$ and $\vec{x}\sim \p^{\otimes \rsa}$. Then
    \[
        \ex{}{Z} = \sa \ab \norm{\p-\bU}_2^2
    \]
    In particular, (1)~if $\p=\bU$ then $\ex{}{Z}=0$, and (2)~if $\totalvardist{\p}{\bU}>\dst$ then $\ex{}{Z} > 4\sa\dst^2$.
\end{lemma}
\begin{proof}
Fix any $j\in[\ab]$. Define the ``noise scale'' $\ell_j \eqdef -\rsa + |\{ i : y_i \in \{(j,0),(j,1)\}\}|$, and note that $\ell_j \sim \Pois(\lambda)$.
We have $N_j(\vec{y}) = c_j(\vec{x}) + \eta_j$, where $\eta_j\sim \Bin(\ell_j, \frac{1}{2})$ is independent of $c_j(\vec{x})$, and $c_j(\vec{x})\sim \Pois(\sa \p_j)$. %
By Lemma \ref{Poisson Mechanism} we get that $\eta_j \sim \Pois(\frac{1}{2}\lambda)$, and therefore overall that $N_j(\vec{y})\sim \Pois(\sa \p_j+\frac{1}{2}\lambda)$.
Moreover, $N_1(\vec{y}),\dots,N_\ab(\vec{y})$ are independent, since the $\ell_j,c_j(\vec{x})$'s are.
It is a simple matter to verify that, for a Poisson random variable,
\[
    \ex{M\sim\Pois(a)}{ (M - \mu)^2 - M  } = (a-\mu)^2\,;
\]
from which, by linearity of expectation, 
\[
        \ex{}{Z} = \frac{\ab}{\sa} \sum_{j=1}^\ab \mleft( \sa \p_j+\frac{1}{2}\lambda -\mu \mright)^2
        = \sa\ab \sum_{j=1}^\ab \mleft( \p_j - \frac{1}{\ab} \mright)^2
        = \sa \ab \norm{\p-\bU}_2^2\,,
\]
as claimed, where we used our definition of $\mu = \frac{\sa}{\ab} + \frac{\lambda}{2}$.
To conclude, note that if $\p=\bU$ then $\norm{\p-\bU}_2^2 = 0$; while if $\tv{\p}{\bU} > \dst$ then
 \[
    \ab\norm{\p-\bU}_2^2 
    \geq \norm{\p-\bU}_1^2
    = 4\totalvardist{\p}{\bU}^2 > 4\dst^2 \,,
\]
the first inequality being Cauchy--Schwarz.
\end{proof}
\begin{lemma}[Variance in the uniform case]
\label{var of Z when uniform}
    Let $Z$ be defined as in~\eqref{eq:def:statistic:Z}, when $\rsa\sim\Pois(\sa)$ and $\vec{x}\sim \bU^{\otimes \rsa}$. Then
    \[
        \Var{}{Z} = \frac{2\ab^3}{\sa^2}\cdot \mu^2
    \]
\end{lemma}
\begin{proof}
By independence of the $N_1(\vec{y}),\dots,N_\ab(\vec{y})$'s, we have
\[
    \Var{}{Z} = \frac{\ab^2}{\sa^2}\sum_{j=1}^\ab\Var{}{\mleft( (N_j(\vec{y}) - \mu)^2 - N_j(\vec{y}) \mright)} = \frac{\ab^3}{\sa^2}  \Var{}{\mleft( (N_1(\vec{y}) - \mu)^2 - N_1(\vec{y}) \mright)} 
\]
the second equality since, in the uniform case, all $N_j(\vec{y})$'s are i.i.d.\ Poisson with parameter $\mu$. One can again check that 
\[
    \ex{M\sim\Pois(\mu)}{ \mleft( (M - \mu)^2 - M \mright)^2 } = 2\mu^2\,;
\]
proving the result.
\end{proof}
We finally provide the last piece of the puzzle, bounding the variance in the non-uniform case. We note that while~\cref{var of Z when uniform} directly follows from the general case below (plugging in $\ex{}{Z}=0$), its proof illustrated the key ideas and the simplicity of the approach, hence our choice to make it a standalone lemma.
\begin{lemma}[Variance in the ``far'' case]
\label{var of Z when far}
    Let $Z$ be defined as in~\eqref{eq:def:statistic:Z}, when $\rsa\sim\Pois(\sa)$ and $\vec{x}\sim \p^{\otimes \rsa}$. Then, 
    \[
        \Var{}{Z} \leq \frac{2\ab^3\mu^2}{\sa^2} + \Paren{ \frac{2\ab}{\sa}+\frac{4\ab^{3/2}\mu}{\sa} } \ex{}{Z} + \frac{4\ab^{1/2}}{\sa^{1/2}} \ex{}{Z}^{3/2}
    \]
\end{lemma}
The statement follows from computations analogous to those from~\cite[Appendices~A,B]{AcharyaDK15}, involving (again) the moments of Poisson random variables along with the Cauchy--Schwarz inequality and monotonicity of $\ell_p$ norms. In the interest of conciseness, we provide it in~\cref{app:technical}.

With those lemmas in hand, we can bound the sample complexity as a function of the various parameters:
\begin{lemma}
\label{bound for n}
There exists $C>0$ such that the following holds. If $\sa \geq C\cdot \frac{\ab^{3/4}\mu^{1/2}}{\dst}$, then $\proba{Z \geq \tau}\leq \frac{1}{3}$ when $\p=\bU$ and $\proba{ Z \leq \tau} \leq \frac{1}{3}$ when $\totalvardist{\p}{\bU} > \dst$. (Moreover, one can take $C=40$.)
\end{lemma}
\begin{proof}
Consider first the case $\p=\bU$. Then, by~\cref{mean of Z,var of Z when uniform} and Chebyshev's inequality, we have that 
\[
    \proba{Z > \tau} \leq \frac{\Var{}{Z}}{\tau^2} = \frac{2\ab^3\mu^2}{\sa^2\cdot (2\sa\dst)^2} \leq \frac{1}{3}
\]
the last inequality as long as $\sa \geq \sqrt[4]{\frac{3}{2}}\cdot \frac{\ab^{3/4}\mu^{1/2}}{\dst}$.

Turning to the case $\totalvardist{\p}{\bU} > \dst$, we have $\ex{}{Z} > 2\tau$ by~\cref{mean of Z}. Thus, from~\cref{var of Z when far} and Chebyshev,
\begin{align*}
    \proba{Z \leq \tau} \leq \proba{|Z-\ex{}{Z}| \geq \frac{1}{2}\ex{}{Z}}
    &\leq \frac{4\Var{}{Z}}{\ex{}{Z}^2} \\
    &\leq \frac{8\ab^3\mu^2}{\sa^2\ex{}{Z}^2} 
    + 4\Paren{ \frac{2\ab}{\sa}+\frac{4\ab^{3/2}\mu}{\sa}}\frac{\ex{}{Z}}{\ex{}{Z}^2}
    + 16\frac{\ab^{1/2}}{\sa^{1/2}}\frac{\ex{}{Z}^{3/2}}{\ex{}{Z}^{2}}\\
    &\leq \frac{2\ab^3\mu^2}{\sa^4\dst^4} 
    + \frac{4\ab}{\sa^2\dst^2}+\frac{8\ab^{3/2}\mu}{\sa^2\dst^2}
    + \frac{16\ab^{1/2}}{\sqrt{2}\sa\dst}\\
    &\leq \frac{2\ab^3\mu^2}{\sa^4\dst^4} 
    + \frac{12\ab^{3/2}\mu}{\sa^2\dst^2}
    + \frac{12\ab^{3/4}\mu^{1/2}}{\sa\dst}
\end{align*}
where in the last inequality we used $\mu\geq 1$ and $16/\sqrt{2}\leq 12$ to slightly loosen the upper bound on the second and fourth term, for convenience. Indeed, setting $t \eqdef \frac{\ab^{3/4}\mu^{1/2}}{\sa\dst}$, the RHS is now $2t^4+12t^2+12t$, which after some calculus is less than $1/3$ for $0\leq t \leq 1/40$. Thus, having $\sa \geq 40\cdot \frac{\ab^{3/4}\mu^{1/2}}{\dst}$ suffices to ensure $\proba{Z \leq \tau} \leq 1/3$ in this case. Taking $C = \max(40,\sqrt[4]{3/2})=40$ concludes the proof.
\end{proof}
To conclude the proof of the theorem, it only remains to substitute our choice of $\mu= \frac{\sa}{\ab}+\frac{\lambda}{2}$ and $\lambda=O\Paren{\log(1/\delta)/\priv^2}$ (\cref{eq:lambda}) in the bound of~\cref{bound for n}. This leads to the sufficient condition
\[
 \sa \geq C'\cdot \frac{\ab^{3/4}}{\alpha}\Paren{ \frac{\sa^{1/2}}{\ab^{1/2}} + \frac{\sqrt{\log(1/\delta)}}{\priv} }
\]
for some constant $C'>0$, yielding the claimed sample complexity $\sa  = O\Paren{ \frac{\ab^{3/4}}{\dst \priv} \log^{1/2}\frac{1}{\delta} + \frac{\ab^{1/2}}{\dst^2}  }$.
\end{proof}

%
\section{Second algorithm: from local to shuffle privacy}
\label{4}
We turn to our second algorithm for shuffle private uniformity testing, which will enjoy the same guarantees as the first, and be conceptually remarkably simple. Indeed, the gist of the algorithm is as follows: (1)~take a \emph{locally private} uniformity testing algorithm with optimal sample complexity $O(\ab^{3/2}/(\dst^2e^\priv))$ in the \emph{low privacy} regime ($\priv \gg 1$); (2)~make all $\sa$ users use this randomiser, with privacy parameter set to $\priv_{L}\approx\log(\priv^2\sa/\log(1/\delta))$. Then, a recent result of~\cite{FeldmanMT20} on ``privacy amplification via shuffling'' will guarantee that the resulting shuffle protocol is $\priv$-DP, and showing that the resulting sample complexity matches~\eqref{eq:sc:privatecoin:better} will only take a few lines.

Unfortunately for this plan, the first ingredient (1) is missing from the existing literature; we thus have to provide it ourselves. The resulting optimal LDP algorithm is itself quite simple, both conceptually and algorithmically; and we believe it to be of independent interest. The reader will find its description and analysis in~\cref{ssec:4ldp}; we for now assume (1), and show in~\cref{ssec:4shuffle} how to use it to obtain a shuffle private tester.

\subsection{Robust shuffle privacy by privacy amplification}
    \label{ssec:4shuffle}
    
We begin by stating the key result our analysis will rely on, here instantiated to our setting of non-interactive, private-coin protocols.
\begin{theorem}[{\cite[Theorem~3.1]{FeldmanMT20}}]\label{higheps0}
 For any domain $\cX$, let $\cR\colon\cX\to\cY$ be an $\priv_{L}$-DP local randomiser; and let $\cS$ be the algorithm that given a tuple of $\sa$ messages $\vec{y}\in\cY^\sa$, samples a uniform random permutation $\pi$ over $[\sa]$ and outputs $(\vec{y}_{\pi(1)},\dots,\vec{y}_{\pi(\sa)})$.
 Then for any $\delta\in(0,1]$ such that $\priv_{L}\le\log\frac{\sa}{16\log(2/\delta)}$, $\cS\circ\cR^\sa$ is $(\priv, \delta)$-DP,  where
\[
\priv\le \log\Paren{1+8\frac{e^{\priv_{L}}-1}{e^{\priv_{L}}+1}\left(\sqrt{\frac{e^{\priv_{L}}\log(4/\delta)}{\sa}}+\frac{e^{\priv_{L}}}{n}\right) }\,.
\]
In particular, if $\priv_{L}\geq 1$ then $\priv = O\Paren{\sqrt{e^{\priv_{L}}\log(1/\delta)/\sa}}$, and if $\priv_{L} < 1$ then $\priv = O\Paren{\priv_{L}\sqrt{\log(1/\delta)/\sa}}$.
\end{theorem}
We also remark that the analysis implies that the resulting protocol is actually $(\priv,4\delta^\gamma,\gamma)$-\emph{robustly} shuffle private, as for $\gamma\sa$ honest users it inherits the same amplification guarantees, replacing $\sa$ by $\gamma\sa$ in the privacy bound. We detail this in the proof below. 

With this at our disposal, we are ready to state and prove the main theorem of this section.
\begin{theorem}
\label{thm:bound:via:shuffling}
Let $\gamma \in (0,1]$, $\priv \in(0,1]$, and $\dst, \delta \in (0,1)$; and denote by $\cR_{\rm ldp}$ and $\cA_{\rm ldp}$ the locally private randomiser and analyser from~\cref{theo:newldp}. Then, the \emph{shuffle} protocol $\cP = (\cR_{\rm ldp},\cA_{\rm ldp})$ is $(\priv, 4\delta^\gamma,\gamma)$-robustly shuffle private. Then $\cP$ solves $\dst$-uniformity testing with sample complexity
	\[
	  \sa  = O\Paren{ \frac{\ab^{3/4}}{\dst \priv} \log^{1/2}\frac{1}{\delta} + \frac{\ab^{1/2}}{\dst^2}  }
	\]
Moreover, each user sends only one message, using $\log_2\ab+O(1)$ bits of communication.
\end{theorem}
\begin{proof}
Let $\priv_{L}>0$ be defined as the (unique) positive solution of 
\begin{equation}
    \label{eq:choice:epsL}
    \priv = \log\Paren{ 1+16e^{\priv_{L}/2}\cdot\frac{e^{\priv_{L}}-1}{e^{\priv_{L}}+1}\sqrt{\frac{\log(4/\delta)}{\sa}} }\,.
\end{equation}
One can check that this satisfies the assumption of~\cref{higheps0}; therefore, applying~\cref{higheps0} to $\cP = (\cR_{\rm ldp},\cA_{\rm ldp})$ with our choice of $\priv_L$, we obtain the $(\priv,\delta)$-shuffle privacy we sought (and, as per the remark above, the robust shuffle privacy claimed, as~\eqref{eq:choice:epsL} remains unchanged when replacing both $\sa$ and $\delta$ by $\gamma\sa$ and $\delta(\gamma) \eqdef 4^{1-\gamma}\delta^\gamma \leq 4\delta^\gamma$).

Turning to the utility, we first rewrite~\eqref{eq:choice:epsL}
as
\[
e^{\priv_{L}/2}\cdot\frac{e^{\priv_{L}}-1}{e^{\priv_{L}}+1} = \frac{e^\priv-1}{16} \cdot \sqrt{\frac{\sa}{\log(4/\delta)}}
= \Theta\Paren{ \priv\sqrt{\frac{\sa}{\log(1/\delta)} } }
\]
so that, given the inequalities $\frac{e^x+1}{e^x-1}\geq 1$, $\frac{e^x+1}{e^{x/2}}\geq 2$ for all $x > 0$, we get that both
$
e^{\priv_{L}}=\Omega\Paren{\frac{\priv^2\sa}{\log(1/\delta)}}$ and $
(e^{\priv_{L}}-1)^2 = \Omega\Paren{\frac{\priv^2\sa}{\log(1/\delta)}}
$. 
Now, we know from~\cref{theo:newldp} that for the resulting algorithm to solve $\dst$-uniformity testing, it suffices to have
\[
    \sa = O\Paren{ \frac{\ab^{3/2}}{\dst^2(e^{\priv_{L}}-1)^2} + \frac{\ab^{3/2}}{\dst^2e^{\priv_{L}}} + \frac{\ab^{1/2}}{\dst^2}}
\]
which, by the above discussion, leads to the sufficient condition
$
    \sa = O\Paren{ \frac{\ab^{3/2}\log(1/\delta)}{\dst^2\priv^2\sa} + \frac{\ab^{1/2}}{\dst^2}}
$. Reorganising this expression and solving for $\sa$ yields the claimed sample complexity.
\end{proof}
\subsection{A locally private tester in the low privacy regime}
    \label{ssec:4ldp}

In this section, we provide the last missing piece, and establish the following theorem on locally private uniformity testing:
\begin{theorem}
    \label{theo:newldp}
Let $\priv>0$ and $\dst\in(0,1]$. There exists a private-coin $\priv$-locally private protocol (\cref{algo:ldp:local:randomiser,algo:ldp:server:side}) which solves $\dst$-uniformity testing with sample complexity
\[
    \sa = 
    O\Paren{ \frac{\ab^{3/2}}{\dst^2(e^\priv-1)^2} + \frac{\ab^{3/2}}{\dst^2e^\priv} + \frac{\ab^{1/2}}{\dst^2}}\,,
\]
where each user sends $\log_2\ab + O(1)$ bits. Moreover, this sample complexity is optimal among all private-coin $\priv$-locally private protocols.
\end{theorem}
Before turning to the proof, we discuss a few aspects of the statement. The first term dominates when $\priv = O(1)$, and it's this regime where sample complexity is $O(\ab^{3/2}/(\dst^2\priv^2))$, known from~\cite{AcharyaCFT19,AcharyaCFST21} (and shown optimal in~\cite{AcharyaCT20a}). The last term dominates when $\priv > \log\ab$, for which we retrieve the optimal sample complexity from the non-private case. Thus, the novelty of our result lies in the intermediate, ``low-privacy'' regime $1\leq \priv \leq \log\ab$, which had not previously been considered but will be crucial for our application to shuffle privacy in the previous subsection. Moreover, by using the domain compression technique, this result also readily yields an optimal locally private \emph{public-coin} tester, which was before only known in the high-privacy regime. We detail this in~\cref{public-coin bounds}.

\begin{algorithm}[ht]
\begin{algorithmic}[2]
    \Require Sample $x\in[\ab]$; parameters $\ab\geq 2,\priv>0$
    \State Set $a,b,\outab,s$ as in~\eqref{eq:scheme:ldp:choice:parameters}, and matrix $\bar{H}_{a,b}$ as in~\eqref{eq:scheme:ldp:ghr:matrix} \Comment{Computationally efficient}
    \State Use them to sample $y\in[\outab]$ according to the probabilities given in~\eqref{eq:scheme:transition:proba}
    \State \Return $y$ \Comment{$\clg{\log_2 \outab} \leq \log_2 \ab + 4$ bits}
\end{algorithmic}
\caption{\label{algo:ldp:local:randomiser}Local randomiser}
\end{algorithm}
\begin{algorithm}[ht]
\begin{algorithmic}[2]
    \Require Parameters $\ab\geq 2,\priv>0$, distance parameter $\dst\in(0,1]$
    \Require Vector of $n$ messages $\vec{y}\in[\outab]^n$ from the $\sa$ local randomisers (\cref{algo:ldp:local:randomiser})
    \State Set $\outab,s$ as in~\eqref{eq:scheme:ldp:choice:parameters}, and
\[
        \gamma^2 \gets \frac{2\dst^2}{s\ab}\Paren{\frac{e^\priv-1}{e^\priv + \frac{\outab}{s} - 1}}^2
\]
    \State Compute the probability distribution $\q^\ast = \Phi_{\priv}(\bU_{\ab})$ over $[\outab]$ induced by the mechanism (cf.~\eqref{eq:induced:distribution:hr})
    \State Use the $\ell_2$ testing algorithm of~\cite{ChanDVV14} (cf.~\cite[Theorem~13]{AcharyaCFT19}) to distinguish between $\q=\q^\ast$ and $\norm{\q-\q^\ast}_2 > \gamma$, where $\q$ denotes the distribution of the $n$ messages $\vec{y}_1,\dots,\vec{y}_\sa$.
    \State \Return the output of the $\ell_2$ testing algorithm
\end{algorithmic}
\caption{\label{algo:ldp:server:side}Analyser}
\end{algorithm}

\paragraph{Outline.} The high-level idea of the proof is quite simple, and inspired from the LDP testing protocol of~\cite{AcharyaCFT19} (which was specific to the high-privacy regime)\footnote{Our starting point is the LDP protocol of~\cite{AcharyaCFT19}, instead of that of~\cite{AcharyaCFST21} which is also based on Hadamard Response (HR) and achieves the same optimal sample complexity but is more communication-efficient, for two reasons. First, the former naturally lends to the extension from HR to Generalised HR, while this extension is not obvious for the latter. Second, to apply the privacy amplification by shuffling paradigm in our previous section, we required all $\sa$ users to use the same randomiser; this is the case with the protocol of~\cite{AcharyaCFT19}, but not with that of~\cite{AcharyaCFST21}.} combined with the generalisation of Hadamard Response from~\cite{AcharyaSZ19} -- an LDP mechanism with output space $[\outab]$, for some $\outab=\Theta(\ab)$ depending on $\ab$ and $\priv$. 
Each of the $\sa$ users will apply this mechanism to their sample $x\sim\p$, thus obtaining a sample $y\in[\outab]$; note that the resulting $y$ is then, over the whole process,  drawn from some probability distribution $\Phi(\p)$ over $[\outab]$ induced by $\p$. The analyser thus receives $\sa$ i.i.d.\ ``privatised'' samples from this induced distribution $\Phi(\p)$, and uses them to test (in $\ell_2$ distance, not total variation) if $\Phi(\p)=\Phi(\bU)$, which is what it should be if $\p$ were uniform; this testing itself can be done using a known (non-private) testing algorithm. Details follow.

\paragraph{Generalised Hadamard Response.} We start by recalling the ``generalised Hadamard Response'' (GHR) mechanism from~\cite{AcharyaSZ19}. Consider the following (family of) schemes, parametrised by $\priv>0$ and integers $\ab,s,\outab\in\N$ with $s\leq \outab$: the randomiser $\cR\colon[\ab]\to[K]$ is given by
\begin{equation}
    \label{eq:scheme:transition:proba}
   \proba{ \cR(x) = y } \eqdef \frac{(e^\priv-1)\ind{y\in C_x} + 1}{s e^\priv + \outab - s}
\end{equation}
where, for every $x\in[\ab]$, $C_x\subseteq [\outab]$ with size $|C_x|=s$. It is immediate to see from~\eqref{eq:scheme:transition:proba} that
\begin{equation}
    \label{eq:scheme:ldp}
    \sup_{y\in[\outab]}\sup_{x,x'\in[\ab]} \frac{\proba{ \cR(x) = y }}{\proba{ \cR(x') = y }} \leq e^\priv
\end{equation}
so the mechanism above is $\priv$-LDP for any choice of $\ab,\outab,s$. Given an input alphabet size $\ab$ and a privacy parameter $\priv>0$, the above definition leaves us the choice of $\outab$, $s$, and of the collection of subsets $(C_x)_{x\in[\ab]}$. Those will be obtained from properties of the Hadamard matrices.  
Specifically, given $\ab\in\N$ and $\priv>0$, define
\begin{equation}
    \label{eq:scheme:ldp:choice:parameters}
    a \eqdef 2^{\flr{\log_2\min(e^\priv, 2\ab)}}%
    , \quad
    b \eqdef 2^{\clg{\log_2(\frac{\ab}{a}+1)}}, \quad 
    \outab \eqdef a\cdot b, \quad 
    s \eqdef \frac{b}{2} \geq 1.
\end{equation}
One can check that (1)~$\outab \leq 2(\ab+a) \leq 6\ab$, (2)~$\outab \geq \ab + a$, and of course (3)~$a,b,\outab,s$ are powers of two. We then consider the \emph{$(a,b)$-reduced Hadamard matrix} $\bar{H}_{a,b}\in \{-1,1\}^{\outab\times \outab}$ defined as
\begin{equation}
    \label{eq:scheme:ldp:ghr:matrix}
        \bar{H}_{a,b} \eqdef 
        \begin{pmatrix}
        H_b & P_b & \cdots & P_b \\
        P_b & H_b & \cdots & P_b \\
        \vdots &\vdots & \ddots & \vdots \\
        P_b & P_b & \cdots & H_b \\
        \end{pmatrix}
\end{equation}
where $H_b$ is the $b\times b$ Hadamard matrix (Sylvester's construction), and $P_b$ is the $b\times b$ matrix with all entries equal to $-1$.  (Each row and column of $\bar{H}_{a,b}$ has $a-1$ occurrences of $P_b$, and one of $H_b$.)\smallskip

In particular, since $H_b$'s first row is all-ones, and each of the remaining $b-1$ has exactly $s=b/2$ entries set to $+1$, the matrix $\bar{H}_{a,b}$ has $a\cdot (b-1) = \outab-a \geq \ab$ rows with exactly $s$ entries equal to $+1$ (and $s$ equal to $-1$). This allows us to map each $x\in [\ab]$ to a distinct row $\vec{r}_x$ of $\bar{H}_{a,b}$, which we can see naturally as a subset $C_x\subseteq [\outab]$ of size $|C_x|=s$ (by interpreting the row $\vec{r}_x\in\{\pm 1\}^\outab$ as the indicator vector of $C_x$).

We now have all the elements to prove~\cref{theo:newldp}. For simplicity of exposition, the proof of some of the technical lemmas, which are essentially computations and are either analogous to those from~\cite{AcharyaSZ19} or follow from properties of the Hadamard matrix, are deferred to~\cref{app:technical:ldp}.
\begin{proof}[Proof of~\cref{theo:newldp}]
For a distribution $\p$ on $[\ab]$, denote by $\Phi_{\priv}(\p)$ the probability distribution on $[\outab]$ obtained by applying the mechanism defined in~\eqref{eq:scheme:transition:proba}, with parameters~\eqref{eq:scheme:ldp:choice:parameters}, on an input $x\sim\p$. In particular, we have, for any $y\in[\outab]$,
\begin{align}
    \label{eq:induced:distribution:hr}
    \Phi_{\priv}(\p)(y) 
    &= \sum_{x\in[\ab]} \p(x)\proba{\cR(x)=y} 
    = \frac{1}{s e^\priv + \outab - s} \Paren{ (e^\priv-1)\sum_{x\in[\ab]} \p(x)\ind{y\in C_x} + 1 }
\end{align}
We then have the following ``Parseval-type'' statement, which guarantees that the mapping $\Phi_{\priv}$ preserves the $\ell_2$ distances between the original distributions.
\begin{lemma}
    \label{lemma:parseval:like:ghr}
Let $\p,\q$ be two arbitrary distributions over $[\ab]$. Then
\[
    \norm{\Phi_{\priv}(\p)-\Phi_{\priv}(\q)}_2^2 \geq \frac{1}{2s}\Paren{\frac{e^\priv-1}{e^\priv + \frac{\outab}{s} - 1}}^2 \norm{\p-\q}_2^2
\]
Moreover, if $\p=\q$ then of course $\norm{\Phi_{\priv}(\p)-\Phi_{\priv}(\q)}_2=0$.
\end{lemma}
Based on the above, the server receives $\sa$ i.i.d.\ samples from a distribution $\q \eqdef \Phi_\priv(\p)$ over $[\outab]$, such that, letting $\q^\ast \eqdef \Phi_\priv(\bU_{\ab})$ be the reference (known) probability distribution, we either have
(1)~$\q = \q^\ast$
or
(2)~$
\norm{\q-\q^\ast}_2^2 > \frac{1}{2s}\Paren{\frac{e^\priv-1}{e^\priv + {\outab}/{s} - 1}}^2\cdot \frac{4\dst^2}{\ab} \eqdef \gamma^2.
$
It is known~\cite{ChanDVV14} (or~\cite[Theorem~13]{AcharyaCFT19}) that to distinguish between (1) and (2), it suffices to have
$
    \sa = O\Paren{{\norm{\q^\ast}_2}/{\gamma^2}}\,.
$ 
To conclude, it thus suffices to show that $\norm{\q^\ast}_2$ is suitably small, which we do in the next lemma.
\begin{lemma}
    \label{lemma:induced:reference:small:l2}
With the above notation $\q^\ast \eqdef \Phi_\priv(\bU_{\ab})$, we have 
$
    \norm{\q^\ast}_2^2 = O\Paren{\frac{1}{\outab}}.
$
\end{lemma}
In essence, this states that the $\ell_2$ norm of our reference $\q^\ast$ is (up to constant factors) as small as it gets, since the minimum possible value for the $\ell_2$ norm of a distribution over $\outab$ elements is $1/\sqrt{\outab}$. Recalling the value of $\gamma$, using this lemma we obtain
\[
\frac{\norm{\q^\ast}_2}{\gamma^2}
\leq 2s\Paren{\frac{e^\priv + \frac{\outab}{s} - 1}{e^\priv-1}}^2\cdot \frac{\ab}{4\dst^2}\cdot \sqrt{\frac{24}{\outab}}
= O\Paren{ \frac{\ab^{3/2}}{\dst^2}\cdot \frac{s}{\ab}\cdot \Paren{\frac{e^\priv}{e^\priv-1}}^2 }
= O\Paren{ \frac{\ab^{3/2}}{\dst^2}\cdot \frac{e^\priv}{\Paren{e^\priv-1}^2} + \frac{\sqrt{\ab}}{\dst^2} }
\]
the last equality since $s/\ab = O\Paren{ 1/a } = O\Paren{e^{-\priv}+ 1/\ab }$. Noting that $
\frac{e^\priv}{\Paren{e^\priv-1}^2} = O\Paren{\frac{1}{\Paren{e^\priv-1}^2} + \frac{1}{e^\priv}}
$ concludes the proof of the upper bound.\smallskip

To prove the lower bound, it suffices to establish the middle term, as the other two follow from the known high-privacy and non-private regimes. We will invoke for this~\cite[Corollary~IV.20]{AcharyaCT20a}; to establish the $\Omega\Paren{\frac{\ab^{3/2}}{\alpha^2 e^\priv}}$ lower bound, it suffices to show (with their notation) that, for every $\priv$-LDP randomiser $\cR'\colon[\ab]\to\cY$, the trace of the $\ab/2\times\ab/2$ p.s.d.\ matrix $H(\cR')$ 
given by
\[
        H(\cR')_{i,j} = \sum_{y\in\cY} \frac{(\proba{\cR'(2i)=y}-\proba{\cR'(2i-1)=y})(\proba{\cR'(2j)=y}-\proba{\cR'(2j-1)=y})}{\sum_{x\in[\ab]} \proba{\cR'(x)=y}}\;,\; i,j\in[\ab/2]
\]
(cf. \cite[Definition~I.5]{AcharyaCT20a}) 
is $O(e^\priv)$. This in turn readily follows from the inequality
\begin{align}
\operatorname{Tr}[H(\cR')] 
&\eqdef \sum_{y\in\cY} \frac{\sum_{i=1}^{\ab/2} (\proba{\cR'(2i)=y}-\proba{\cR'(2i-1)=y})^2}{\sum_{x\in[\ab]} \proba{\cR'(x)=y}}
\leq 2\sum_{y\in\cY} \frac{\sum_{x\in[\ab]} \proba{\cR'(x)=y}^2}{\sum_{x\in[\ab]}  \proba{\cR'(x)=y}} \notag\\
&\leq 2\sum_{y\in\cY} e^\priv \min_{x\in[\ab]}\proba{\cR'(x)=y} \cdot\frac{\sum_{x\in[\ab]} \proba{\cR'(x)=y}}{\sum_{x\in[\ab]}  \proba{\cR'(x)=y}} \notag\\
&\leq 2e^\priv \min_{x\in[\ab]}\sum_{y\in\cY} \proba{\cR'(x)=y} 
= 2e^\priv\,, \label{eq:lb:ldp}
\end{align}
where the first equality is by definition of $H(\cR')$, the first inequality is $(a-b)^2 \leq 2(a^2+b^2)$, the second inequality follows from the fact that $\cR'$ is $\priv$-LDP, the third inequality is $\sum \min \leq \min\sum$, and the last equality uses that the probabilities sum to one.\footnote{This simple argument of an $O(e^\priv)$ bound is nearly identical to that of~\cite{AcharyaCST21} for an analogous quantity; we reproduce it here for completeness.} This, along with~\cite[Corollary~IV.20]{AcharyaCT20a}, proves the claimed lower bound.
\end{proof}

\section{Conclusion}
Currently, the only known sample complexity lower bounds for shuffle private uniformity testing hold for \emph{pure} robustly shuffle private protocols~\cite{BalcerCJM21}. As a result, our upper bound is not directly comparable to this lower bound, as our protocol  (and, indeed, all currently known such protocols) only satisfies approximate differential privacy.
This naturally leads to the following questions: 
\begin{itemize}
    \item \emph{What upper bounds can we derive for pure shuffle DP?} This appears to require significantly new ideas, as neither of the algorithms described in this work lends itself to pure differential privacy: the first, as it would require a non-trivial improvement of existing pure shuffle DP algorithms for binary summation; and the second, as amplification by shuffling inherently leads to approximate DP guarantees. A natural idea would be to use a pure shuffle DP algorithm for binary summation~\cite{GhaziGKMPV20}, instead of the randomiser from~\cref{algo:1}. Unfortunately, the distribution of the resulting noise (i.e., the combination of the noise from the sampling of the input themselves, and the added privacy-preserving noise) no longer has a convenient form, which makes the result challenging to analyse.
    \item \emph{What lower bounds can we derive for approximate DP?} The currently known lower bound relies on a reduction from the so-called pan-privacy setting to robust shuffle privacy. One way to extend the lower bound would be to obtain the analogous approximate DP lower bound in the pan-privacy setting; in~\cref{app:lowerbound}, we describe another (simple) lower bound, obtained via a reduction from local privacy, which while relatively weak strongly suggests that our upper bound for approximate shuffle DP is tight. Can tight lower bounds for the shuffle DP setting be obtained directly, without a reduction from either the pan- or locally private settings?
\end{itemize}

\printbibliography

\appendix
\section{Some (partial) lower bounds}\label{app:lowerbound}

In this appendix, we establish two lower bounds for shuffle private uniformity testing with private-coin protocols. The first only applies to robustly shuffle private protocols, and directly follows from the public-coin lower bound of~\cite{BalcerCJM21}:
\begin{theorem}
For $\priv = O(1)$ and $\dst \in (0,1/9)$ such that $\ab \geq 2^{3/2}\dst^2/\priv^2$, any private-coin $(\priv,0,1/3)$-robustly shuffle private protocol $\dst$-uniformity tester must have sample complexity
\[
    \sa = \Omega\Paren{\frac{\ab^{3/4}}{\dst\priv} + \frac{\ab^{1/2}}{\dst^2}}\,.
\]
\end{theorem}
\begin{proof}
The second term of the expression is the non-private lower bound, and thus holds regardless of the privacy parameter. Since the first term only dominates when $\ab \geq (\priv/\dst)^4$, we can assume we are in this parameter regime. Suppose by contradiction there exists a private-coin $(\priv,0,1/3)$-robustly shuffle private protocol uniformity tester with sample complexity 
$
o\Paren{\ab^{3/4}/(\dst\priv)+\ab^{1/2}/\dst^2}
$. Then, setting 
\[
    L \eqdef \clg{\frac{\ab^{2/3}\priv^{4/3}}{\dst^{4/3}}} 
\]
which satisfies $2\leq L\leq \ab$ since $\ab \geq 2^{3/2}\dst^2/\priv^2$ and $\ab \geq (\priv/\dst)^4$. We can therefore use the domain compression technique (\cref{theo:random:dct:hashing}) to obtain a \emph{public-coin} $(\priv,0,1/3)$-robustly shuffle private protocol uniformity tester with sample complexity 
\[
    o\Paren{\frac{L^{3/4}}{(\dst\sqrt{L/\ab})\priv}+\frac{L^{1/2}}{(\dst\sqrt{L/\ab})^2}}
    = o\Paren{\frac{\ab^{2/3}}{\dst^{4/3}\priv^{2/3}}}
\]
which violates the lower bound of~\cite[Theorem~4.3]{BalcerCJM21} (i.e.,~\eqref{eq:sc:publiccoin:lb}).
\end{proof}

Our second lower bound, albeit weak (in that it does not yield any non-trivial dependence on the privacy parameter $\priv$ in the high-privacy regime), is new, and applies even to \emph{non-robust}, \emph{approximate} shuffle privacy.
\begin{theorem}
For $\priv > 0$, $\delta\in(0,1]$, and $\dst \in(0,1]$, any \emph{single-message} private-coin $(\priv,\delta)$-shuffle private protocol $\dst$-uniformity tester must have sample complexity
\[
    \sa = \Omega\Paren{\min\Paren{\frac{\ab^{3/4}}{\dst e^{\priv/2}}, \frac{\ab^{3/2}}{\dst^2 \delta 2^\ell}} + \frac{\ab^{1/2}}{\dst^2}}\,.
\]
where $\ell$ is an upper bound on the number of bits sent by each user. In particular, if $\delta \leq \ab^{3/4}/2^\ell$, then the lower bound 
is
$
    \sa = \Omega\Paren{\frac{\ab^{3/4}}{\dst e^{\priv/2}} + \frac{\ab^{1/2}}{\dst^2}}.
$
\end{theorem}
Before proving the theorem, we note that our algorithm from~\cref{4}, for instance, satisfies these assumptions with $\ell = \log_2\ab + O(1)$, and thus the simplified statement (last sentence of the theorem) applies as long as $\delta = O(1/\ab^{1/4})$: a natural setting, as one would typically set $\delta \ll 1/\sa$.
\begin{proof}
The last term again immediately follows from the known non-private sample complexity lower bound. For the first term, we will rely on~\cite[Theorem~26]{CheuSUZZ19}, which states that if a single-message protocol for $\sa$ users is $(\priv,\delta)$-shuffle private, then removing the shuffler yields a $(\priv+\log\sa,\delta)$-\emph{locally} private protocol for $\sa$ users. Moreover, the resulting LDP protocol is private-coin if the original shuffle protocol was.

Assume for now we have a sample complexity lower bound against $(\priv_L,\delta)$-LDP private-coin protocols (using $\ell$ bits per user) for $\dst$-uniformity testing of
\begin{equation}
    \label{eq:lb:ldp:approx}
    \sa = \Omega\Paren{ \frac{\ab^{3/2}}{(e^{\priv_L}+\delta 2^\ell)\dst^2} }
\end{equation}
which holds for all $\priv_L>0$, $\delta\in[0,1]$, and $\dst\in(0,1]$. Then, by the aforementioned result, we get that any shuffle private protocol as in the theorem statement must have sample complexity satisfying
$
    \sa = \Omega\Paren{ \frac{\ab^{3/2}}{(\sa e^{\priv}+\delta 2^\ell)\dst^2} }
$ which, after reorganising, implies the claimed lower bound.

It remains to show that the lower bound~\eqref{eq:lb:ldp:approx} does hold for locally private protocols. This in turn can be easily shown, with an argument nearly identical to the lower bound part of~\cref{theo:newldp}. Namely, for a $(\priv_L,\delta)$-LDP randomiser $\cR$ with output space $\cY$,~\eqref{eq:lb:ldp} becomes
\begin{align*}
\operatorname{Tr}[H(\cR)] 
&\leq 2\sum_{y\in\cY} (e^\priv \min_{x\in[\ab]}\proba{\cR(x)=y}+\delta) \cdot\frac{\sum_{x\in[\ab]} \proba{\cR(x)=y}}{\sum_{x\in[\ab]}  \proba{\cR(x)=y}} \\
&\leq 2\Big(e^\priv \big(\min_{x\in[\ab]}\sum_{y\in\cY} \proba{\cR(x)=y} \big)  + |\cY|\cdot \delta\Big)
= 2(e^\priv+\delta 2^\ell)\,,
\end{align*}
since $|\cY|\leq 2^\ell$ by assumption. Plugging this in~\cite[Corollary~IV.20]{AcharyaCT20a} proves~\eqref{eq:lb:ldp:approx}.
\end{proof}

\section{Omitted proofs from~\cref{3}}\label{app:technical}

In this appendix, we provide the detailed proof of~\cref{var of Z when far}, bounding the variance of the statistic $Z$ defined in~\cref{eq:def:statistic:Z} in the general (non-uniform) case.
\begin{proof}[Proof of~\cref{var of Z when far}]
We first write
\begin{equation}
    Z =\frac{\mu k}{\sa} \sum_{j=1}^k \frac{(N_j(\vec{y}) - \mu)^2 - N_j(\vec{y})}{\mu}
\end{equation}
Then following a similar argument as in~\cite[Appendices~A,B]{AcharyaDK15}, we have 
\begin{align}
    \Var{}{Z}&=\left(\frac{\mu k}{\sa}\right)^2 \sum_{j=1}^k  \left[ 2\frac{(\p_j+\frac{\lambda}{2\sa})^2}{(\frac{\mu}{\sa})^2}+4m \frac{(\p_j+\frac{\lambda}{2\sa})(\p_j-\frac{1}{\ab})^2}{(\frac{\mu}{\sa})^2} \right] \nonumber\\
    &=\ab^2 \sum_{j=1}^k \left[ 2 \Paren{ \p_j+\frac{\lambda}{2\sa} }^2+4\sa \Paren{\p_j+\frac{\lambda}{2\sa}}\Paren{\p_j-\frac{1}{\ab}}^2 \right] \label{far var}
\end{align}
Similarly, we then bound the terms in \eqref{far var} separately. For the first term,
\begin{align}
    \sum_{j=1}^k  2\Paren{\p_j+\frac{\lambda}{2\sa}}^2
    &=2\sum_{j=1}^k \Paren{ \p_j-\frac{1}{\ab}+\frac{\mu}{\sa} }^2 \nonumber\\
    &=2\sum_{j=1}^k (\p_j-\frac{1}{\ab})^2+2k\frac{\mu^2}{\sa^2}+4\frac{\mu}{\sa}\sum_{j=1}^k \Paren{\p_j-\frac{1}{\ab}} \nonumber \\
    &=2 \norm{\p-\bU}_2^2+2k\frac{\mu^2}{\sa^2}\tag{the last term equals 0} \nonumber \\
    &=2\frac{\ex{}{Z}}{\sa k}+2k\frac{\mu^2}{\sa^2} \label{eq:variance:far:term1}
\end{align}
For the second term,
\begin{align}
    \sum_{j=1}^k 4\sa \Paren{\p_j+\frac{\lambda}{2\sa}}\Paren{\p_j-\frac{1}{\ab}}^2
    &\leq 4\sa \Paren{ \sum_{j=1}^k \Paren{\p_j+\frac{\lambda}{2\sa}}^2 }^{1/2} \Paren{ \sum_{j=1}^k \Paren{\p_j-\frac{1}{\ab}}^4 }^{1/2} \tag{Cauchy--Schwarz} \nonumber\\
    &=4\sa \Paren{ \frac{\ex{}{Z}}{\sa k} +k\frac{\mu^2}{\sa^2} }^{1/2} \norm{\p-\bU}_4^2 \tag{by~\eqref{eq:variance:far:term1}} \nonumber\\
    &\leq 4m \Paren{ \frac{\ex{}{Z}^{1/2}}{\sqrt{\sa k}} +\ab^{1/2}\frac{\mu}{\sa} } \norm{\p-\bU}_2^2 \tag{monotonicity of $\ell_p$ norms} \nonumber \\
    &= 4\sa \Paren{ \frac{\ex{}{Z}^{1/2}}{\sqrt{\sa k}}+\ab^{1/2}\frac{\mu}{\sa} } \frac{\ex{}{Z}}{\sa k} \nonumber \\
    &=4\Paren{ \frac{\ex{}{Z}^{3/2}}{\sa^{1/2}\ab^{3/2}}+\frac{\mu \ex{}{Z}}{\sa \ab^{1/2}} }  \label{eq:variance:far:term2}
\end{align}
Combining the two bounds gives
\begin{align}
    \Var{}{Z}&\leq 
    \ab^2
    \Paren{ 2\frac{\ex{}{Z}}{\sa k} +2k\frac{\mu^2}{\sa^2}+ 4\frac{\ex{}{Z}^{3/2}}{\sa^{1/2}\ab^{3/2}}+4\frac{\mu \ex{}{Z}}{\sa \ab^{1/2}}
    } \nonumber \\
    &=\frac{2\ab^3\mu^2}{\sa^2} + \Paren{ \frac{2k}{\sa}+\frac{4\ab^{3/2}\mu}{\sa}} \ex{}{Z} + 4\frac{\ab^{1/2}}{\sa^{1/2}} \ex{}{Z}^{3/2}
\end{align}
concluding the proof.
\end{proof}

\section{Omitted proofs from~\cref{4}}\label{app:technical:ldp}

\begin{proof}[Proof of~\cref{lemma:parseval:like:ghr}]
The last part of the statement is immediate; we focus on proving the inequality. From the expression of $\Phi_\priv(\p)$ in~\eqref{eq:induced:distribution:hr} we can write, for a given $y\in[\outab]$,
\[
\Phi_{\priv}(\p)(y)  - \Phi_{\priv}(\q)(y) 
= \frac{e^\priv-1}{s e^\priv + \outab - s} \sum_{x\in[\ab]} (\p(x)-\q(x) )\ind{y\in C_x}
\]
From this, we have
\begin{align*}
    \Paren{\frac{s e^\priv + \outab - s}{e^\priv-1}}^2 \norm{\Phi_{\priv}(\p)-\Phi_{\priv}(\q)}_2^2 
    &= \sum_{y\in[\outab]} \Paren{ \sum_{x\in[\ab]} (\p(x)-\q(x) )\ind{y\in C_x} }^2 \\
    &= \sum_{y\in[\outab]} \sum_{x\in[\ab]} (\p(x)-\q(x) )^2 \ind{y\in C_x} \\
    &\quad+ \sum_{y\in[\outab]} \sum_{x\neq x'} (\p(x)-\q(x) )(\p(x')-\q(x') ) \ind{y\in C_x\cap C_{x'}} \\
    &= \sum_{x\in[\ab]} (\p(x)-\q(x) )^2\cdot |C_x| + \sum_{x\neq x'} (\p(x)-\q(x) )(\p(x')-\q(x') ) \cdot |C_x\cap C_{x'}| \\
    &= s\cdot  \sum_{x\in[\ab]} (\p(x)-\q(x) )^2+ \sum_{x\neq x'} (\p(x)-\q(x) )(\p(x')-\q(x') ) \cdot  |C_x\cap C_{x'}|
\end{align*}
We have two cases: 
\begin{itemize}
    \item if $x,x'$ correspond to rows $\vec{r}_x$, $\vec{r}_{x'}$ of a \emph{different} copy of the Hadamard matrix $H_b$, then $|C_x\cap C_{x'}|=0$.
    \item if $x,x'$ correspond to distinct rows of the \emph{same} copy of the Hadamard matrix $H_b$, then by properties of the Sylvester construction we have $|C_x\cap C_{x'}|=b/4=s/2$.
\end{itemize}
\begin{align*}
    \Paren{\frac{s e^\priv + \outab - s}{e^\priv-1}}^2 \norm{\Phi_{\priv}(\p)-\Phi_{\priv}(\q)}_2^2 
    &= s \Paren{  \sum_{x\in[\ab]} (\p(x)-\q(x) )^2+\frac{1}{2}\sum_{\substack{x' \neq x\\ x,x' \text{ same  copy}}} (\p(x)-\q(x) )(\p(x')-\q(x') ) } \\
    &= s \Paren{  \frac{1}{2}\sum_{x\in[\ab]} (\p(x)-\q(x) )^2+\frac{1}{2}\sum_{\substack{x,x' \text{ same  copy}}} (\p(x)-\q(x) )(\p(x')-\q(x') ) } \\
    &= \frac{s}{2} \Paren{  \norm{\p-\q}_2^2+\sum_{i=1}^a\sum_{x,x'\in \text{ copy } i} (\p(x)-\q(x) )(\p(x')-\q(x') ) }
\end{align*}
For $1\leq i \leq a$, denote by $T_i \subseteq[\ab]$ the subset of $x$'s mapped to a row of the $i$-th copy of the Hadamard matrix $H_b$. Then,
\begin{align*}
    \sum_{i=1}^a\sum_{x,x'\in T_i} (\p(x)-\q(x) )(\p(x')-\q(x') )
    = \sum_{i=1}^a \Paren{ \p(T_i)^2 -2\p(T_i)\q(T_i) + \q(T_i)^2 } 
    = \sum_{i=1}^a \Paren{ \p(T_i) - \q(T_i) }^2 \geq 0
\end{align*}
and we can conclude.
\end{proof}

\begin{proof}[Proof of~\cref{lemma:induced:reference:small:l2}]
Recalling~\cref{eq:induced:distribution:hr}, we can bound the $\ell_2$ norm of $\q^\ast$ as follows.
\begin{align*}
    (s e^\priv + \outab - s)^2\norm{\q^\ast}_2^2 
    &= \sum_{y\in[\outab]} \Big({ (e^\priv-1)\sum_{x\in[\ab]} \frac{\ind{y\in C_x}}{\ab} + 1 }\Big)^2 \\
    &\leq 2\frac{(e^\priv-1)^2}{\ab^2}\sum_{y\in[\outab]}\Big({\sum_{x\in[\ab]}\ind{y\in C_x}}\Big)^2 + 2\outab \\
    &= 2\frac{(e^\priv-1)^2}{\ab^2}\Big({\sum_{y\in[\outab]}\sum_{x\in[\ab]}\ind{y\in C_x} + \sum_{y\in[\outab]}\sum_{x\neq x'}\ind{y\in C_x\cap C_{x'}}}\Big) + 2\outab \\
    &= 2\frac{(e^\priv-1)^2}{\ab^2}\Big({\sum_{x\in[\ab]}|C_x| + \sum_{\substack{x,x'\in[\ab]\\x\neq x'}}|C_x\cap C_{x'}|}\Big) + 2\outab \\
    &= 2\frac{(e^\priv-1)^2}{\ab^2}\Big({s\ab + \sum_{i=1}^a \sum_{x\neq x' \in T_i}\frac{s}{2}}\Big) + 2\outab \\
    &\leq 2\frac{(e^\priv-1)^2}{\ab^2}\Big({s\ab + \frac{s}{2}\sum_{i=1}^a |T_i|^2}\Big) + 2\outab \\
    &\leq 2\frac{(e^\priv-1)^2}{\ab^2}\Big({s\ab + \frac{s}{2}s\cdot \sum_{i=1}^a |T_i|}\Big) + 2\outab \\
    &= 2\frac{(e^\priv-1)^2}{\ab^2}\Big({s\ab + \frac{1}{2}s^2\ab}\Big) + 2\outab \\
    &\leq 4\frac{(e^\priv-1)^2 s^2}{\ab} + 2\outab
\end{align*}
where the first inequality is (by laziness) $(a+b)^2 \leq 2(a^2+b^2)$, and $T_i$ is defined as in the end of the proof of ~\cref{lemma:parseval:like:ghr} (so that, in particular, $\sum_{i=1}^a |T_i|=\ab$).
We then get the bound
\[
    \norm{\q^\ast}_2^2 
    \leq \frac{4\frac{(e^\priv-1)^2 s^2}{\ab} + 2\outab}{(s e^\priv + \outab - s)^2 }
    \leq \frac{2}{\outab}\cdot \frac{12(e^\priv-1)^2 s^2 + \outab^2}{(s e^\priv + \outab - s)^2 }
    = \frac{24}{\outab}\cdot \frac{(e^\priv-1)^2 + \frac{1}{12}\Paren{\frac{\outab}{s}}^2}{(e^\priv-1 + \frac{\outab}{s})^2 }
    \leq \frac{24}{\outab}\cdot \frac{(e^\priv-1)^2 + \Paren{\frac{\outab}{s}}^2}{(e^\priv-1 + \frac{\outab}{s})^2 }
    \leq \frac{24}{\outab}
\]
as claimed (where we used, in the denominator, $(x+y)^2\geq x^2+y^2$ for $x,y\geq 0$).
\end{proof}

\section{Public-coin LDP tester}\label{public-coin bounds}
We now show how~\cref{theo:newldp}, combined with the domain compression technique (\cref{theo:random:dct:hashing}), leads to an optimal \emph{public-coin} LDP testing algorithm for all values of the privacy parameter $\priv>0$. Prior to our work, such a tester was only known for the high-privacy regime ($\priv = O(1)$).
\begin{theorem}
    \label{theo:newldp:public}
Let $\priv>0$ and $\dst\in(0,1]$. There exists a public-coin $\priv$-locally private protocol which solves $\dst$-uniformity testing with sample complexity
\[
    \sa = 
    O\Paren{ \frac{\ab}{\dst^2(e^\priv-1)^2} + \frac{\ab}{\dst^2e^{\priv/2}} + \frac{\ab^{1/2}}{\dst^2}}\,,
\]
where each user sends $\log_2\ab + O(1)$ bits. Moreover, this sample complexity is optimal among all public-coin $\priv$-locally private protocols.
\end{theorem}
\begin{proof}
The upper bound follows from our private-coin testing algorithm~\cref{theo:newldp}, combined with domain compression (\cref{{theo:random:dct:hashing}}). Specifically, using public randomness the users can reduce the domain from $\ab$ to any choice $2\leq L\leq \ab$, at the price of contracting the distance parameter $\dst$ to $\dst' \asymp \dst\sqrt{L/\ab}$ (and also paying a constant factor in the sample complexity, to amplify the probability of success). From~\cref{theo:newldp} doing so leads to a public-coin testing algorithm with sample complexity
\[
    \sa = 
    O\Paren{ \frac{L^{3/2}}{\dst^2(L/\ab)(e^\priv-1)^2} + \frac{L^{3/2}}{\dst^2(L/\ab)e^{\priv}} + \frac{L^{1/2}}{\dst^2(L/\ab)}}
    = O\Paren{ \frac{\ab L^{1/2}}{\dst^2(e^\priv-1)^2} + \frac{\ab L^{1/2}}{\dst^2e^{\priv}} + \frac{\ab}{\dst^2L^{1/2}}}\,.
\]
Choosing $L = \min(\ab, \clg{e^{\priv}})\in[2,\ab]$ gives the desired bound.

Turning to the lower bound, we observe that it suffices to show a lower bound for the middle term, as the others are known from the high-privacy regime~\cite[Theorem~V.7]{AcharyaCT20a} and the non-private regime, respectively. By~\cite[Corollary~IV.16]{AcharyaCT20a}, to establish this result it suffices to establish that $\norm{H(\cR)}_F = O(e^{\priv/2})$ for any $\priv$-LDP randomizer $\cR$, where $H(\cR)$ is the same p.s.d.\ matrix as in the proof of~\cref{theo:newldp} (but now we need to bound its Frobenius norm instead of its trace). This can be done as follows. For $y\in\cY$ and $i\in[\ab/2]$, let $w_{y,i} \eqdef \proba{\cR(2i)=y}+\proba{\cR(2i-1)=y}$. Then,
\begin{align*}
    \norm{H(\cR)}_F^2 &= \sum_{i,j\in[\frac{\ab}{2}]}\Paren{ \sum_{y\in\cY} \frac{ (\proba{\cR(2i)=y}-\proba{\cR(2i-1)=y})(\proba{\cR(2j)=y}-\proba{\cR(2j-1)=y})}{\sum_{x\in[\ab]} \proba{\cR(x)=y}} }^2\\
    &\leq \sum_{y,y'\in\cY} \frac{\Paren{\sum_{i=1}^{\ab/2} w_{y,i}w_{y',i}}^2}{\sum_{i=1}^{\ab/2} w_{y,i}\sum_{i=1}^{\ab/2} w_{y',i}}
    \leq
    \sum_{y,y'\in\cY} \frac{\Paren{2e^\priv\min_x \proba{\cR(x)=y} } \sum_{i=1}^{\ab/2} w_{y',i} \Paren{\sum_{i=1}^{\ab/2} w_{y,i}w_{y',i}}}{\sum_{i=1}^{\ab/2} w_{y,i}\sum_{i=1}^{\ab/2} w_{y',i}}
    \\
    &= \sum_{y,y'\in\cY} \frac{\Paren{2e^\priv\min_x \proba{\cR(x)=y} } \Paren{\sum_{i=1}^{\ab/2} w_{y,i}w_{y',i}}}{\sum_{i=1}^{\ab/2} w_{y,i}}
    = \sum_{y\in\cY} \frac{\Paren{4e^\priv\min_x \proba{\cR(x)=y} } \Paren{\sum_{i=1}^{\ab/2} w_{y,i}}}{\sum_{i=1}^{\ab/2} w_{y,i}} \\
    &= 4e^\priv \sum_{y\in\cY} \min_x \proba{\cR(x)=y} \leq 4e^\priv \min_x \sum_{y\in\cY} \proba{\cR(x)=y} = 4e^\priv\,,
\end{align*}
where we first bounded each $\proba{\cR(2i)=y}-\proba{\cR(2i-1)=y} \leq w_{y,i}$ and expanded the square, then used the $\priv$-LDP assumption on $\cR$, then simplified using that $\sum_{y'}w_{y',i} = 2$ for all $i$; and finally used that $\sum \min \leq \min \sum$. This establishes that $\norm{H(\cR)}_F \leq 2e^{\priv/2}$, which by~\cite[Corollary~IV.16]{AcharyaCT20a} 
shows a $\Omega\Paren{\frac{\ab}{\dst^2e^{\priv/2}}}$ lower bound on the sample complexity.
\end{proof}

\end{document}

%% file: preamble.tex
\usepackage[utf8]{inputenc}
\usepackage[dvipsnames]{xcolor}
\usepackage[colorlinks,citecolor=blue,urlcolor=blue]{hyperref}
\usepackage{amsmath,amssymb,bbm}
\usepackage{mleftright}
\usepackage{fullpage}

\usepackage[capitalize]{cleveref}

\usepackage{alltt}
\usepackage{amsfonts}
\usepackage{amsmath}
\usepackage{amssymb}
\usepackage{amsthm}
\usepackage{array}
\usepackage{bbold}
\usepackage{booktabs}
\usepackage{caption}
\usepackage{enumitem}
\usepackage{fancyhdr}
\usepackage{graphicx}
\usepackage{mathdots}
\usepackage{mathtools}
\usepackage{microtype}
\usepackage{multirow}
\usepackage{pdflscape}
\usepackage{siunitx}
\usepackage{textcomp}
\usepackage{slashed}
\usepackage{tabularx}
\usepackage{xcolor}

\usepackage{algorithm}
\usepackage[noend]{algpseudocode}

\newcommand{\ab}{k}
\newcommand{\dst}{\alpha}
\newcommand{\priv}{\varepsilon}
\newcommand{\p}{\mathbf{p}}
\newcommand{\q}{\mathbf{q}}
\newcommand{\sa}{n}
\newcommand{\rsa}{N}

\newcommand{\Paren}[1]{\mleft(#1\mright)}

\ifnum\withcomments=1
\usepackage[textsize=tiny]{todonotes}
\newcommand{\cmargin}[1]{\todo[backgroundcolor=white]{\textbf{C:} #1}}
\newcommand{\cnote}[1]{\todo[textcolor=orange!70!black,backgroundcolor=white,inline]{\textbf{C:} #1}}
\newcommand{\hmargin}[1]{\todo[backgroundcolor=white]{\textbf{H:} #1}}
\newcommand{\hnote}[1]{\todo[textcolor=blue,backgroundcolor=white,inline]{\textbf{H:} #1}}
\else
\newcommand{\cmargin}[1]{}
\newcommand{\cnote}[1]{}
\newcommand{\hmargin}[1]{}
\newcommand{\hnote}[1]{}
\fi

\newtheorem{theorem}{Theorem}[section]

\newtheorem{lemma}{Lemma}

\crefname{fact}{Fact}{Facts}
\crefname{lemma}{Lemma}{Lemmas}
\newtheorem{definition}{Definition}

\newcommand{\eqdef}{\coloneqq}
\newcommand{\clg}[1]{\mleft\lceil#1\mright\rceil}
\newcommand{\flr}[1]{\mleft\lfloor#1\mright\rfloor}

\newcommand{\bU}{\mathbf{U}}

\newcommand{\cA}{\mathcal{A}}

\newcommand{\cM}{\mathcal{M}}

\newcommand{\cP}{\mathcal{P}}

\newcommand{\cR}{\mathcal{R}}
\newcommand{\cS}{\mathcal{S}}

\newcommand{\cX}{\mathcal{X}}
\newcommand{\cY}{\mathcal{Y}}
\newcommand{\cZ}{\mathcal{Z}}

\newcommand{\Ber}[1]{\mathbf{Ber}\mleft(#1\mright)}

\newcommand{\Bin}{\mathbf{Bin}}

\newcommand{\Pois}{\mathbf{Poi}}

\newcommand{\A}{\mathcal{A}}
\newcommand{\B}{\mathcal{B}}
\newcommand{\D}{\mathcal{D}}
\renewcommand{\H}{\mathcal{H}}
\newcommand{\I}{\mathcal{I}}

\newcommand{\M}{\mathcal{M}}
\renewcommand{\O}{\mathcal{O}}

\renewcommand{\S}{\mathcal{S}}
\newcommand{\X}{\mathcal{X}}
\newcommand{\Y}{\mathcal{Y}}
\newcommand{\Z}{\mathbb{Z}}

\newcommand{\E}[2]{\mathbb{E}_{#1}\mleft[ #2 \mright]}
\newcommand{\ex}[2]{{\ifx&#1& \mathbb{E} \else \underset{#1}{\mathbb{E}} \fi \mleft[#2\mright]}}
\newcommand{\condex}[3]{{\ifx&#1& \mathbb{E} \else \underset{#1}{\mathbb{E}} \fi \mleft[#2\;\middle|\; #3\mright]}}
\newcommand{\proba}[2][]{{\ifx&#1& \Pr \else \underset{#1}{\Pr} \fi \mleft[#2\mright]}}

\newcommand{\ind}[1]{\mathbbm{1}\mleft[#1\mright]}

\newcommand{\N}{\mathbb{N}}
\newcommand{\norm}[1]{\mleft\vert\mleft\vert #1 \mright\vert\mright\vert}
\renewcommand{\P}[2]{\mathbb{P}_{#1}\mleft[#2\mright]}
\newcommand{\R}{\mathbb{R}}

\newcommand{\tv}[2]{\|#1 - #2\|_\mathrm{TV}}
\newcommand{\ut}{\mathsf{UT}}
\newcommand{\Var}[2]{\operatorname{Var}_{#1}\mleft[#2\mright]}

\newcommand{\totalvardist}[2]{\operatorname{d}_{\rm{}TV}(#1,#2)}